\documentclass[11pt]{article}
\usepackage[titletoc,title]{appendix}
\usepackage{amsmath,amsfonts,bm,amsthm, mathrsfs}
\usepackage[colorlinks,citecolor=blue,linkcolor=blue]{hyperref}
\usepackage{tikz}
\usepackage{simpler-wick}
\usetikzlibrary{shapes,arrows}
\usetikzlibrary{intersections,shapes.arrows}
\usetikzlibrary{decorations.pathreplacing}
\usetikzlibrary{arrows.meta}
\numberwithin{equation}{section}
\newtheorem{theorem}{Theorem}[section]
\newtheorem{lemma}[theorem]{Lemma}
\newtheorem{proposition}[theorem]{Proposition}
\newtheorem{corollary}[theorem]{Corollary}

\theoremstyle{definition}
\newtheorem{example}[theorem]{Example}
\newtheorem{definition}[theorem]{Definition}

\theoremstyle{remark}
\newtheorem{remark}[theorem]{Remark}

\usepackage{multicol}
\usepackage{longtable}

\usepackage[T1]{fontenc}
\usepackage{amssymb}
\let\astar=\star

\let\star=\astar
\usepackage[top=1in, bottom=1in, left=1in, right=1in]{geometry}
\usepackage{fancyhdr}
\newcommand{\wickk}[1]{:\!{#1}\!:}

\usepackage{enumerate}
\usepackage{xcolor,cancel}
\usepackage{relsize}

\allowdisplaybreaks

\usepackage{pifont}
\date{\today}

\begin{document}

\title{Conservation Law and Trace Anomaly for the Stress Energy Tensor of a Self-Interacting Scalar Field}

\author{
Beatrice Costeri$^{\dagger}$\thanks{$^{\dagger}$BC:
		Dipartimento di Fisica,
		Universit\`a degli Studi di Pavia \& INFN and INdAM, Sezione di Pavia,
		Via Bassi 6,
		I-27100 Pavia,
		Italia;
		\mbox{beatrice.costeri01@universitadipavia.it}}
	\and
\underline{Claudio Dappiaggi}$^{\ast}$ \thanks{$^{\ast}$CD: Dipartimento di Fisica,
	Universit\`a degli Studi di Pavia \& INFN and INdAM, Sezione di Pavia, 
	Via Bassi 6,
	I-27100 Pavia,
	Italia;
	\mbox{claudio.dappiaggi@unipv.it}, corresponding author}
\and
Michele Goi$^{\circ}$\thanks{$^{\circ}$MG: Dipartimento di Fisica,
		Universit\`a degli Studi di Pavia,
		Via Bassi 6,
		I-27100 Pavia,
		Italia;
	\mbox{michele.goi01@universitadipavia.it}
}}


\renewcommand\footnotemark{}


\maketitle

\begin{abstract}
We consider a self-interacting, massive, real scalar field on a four-dimensional globally hyperbolic spacetime and the associated stress-energy tensor. Using techniques proper of the algebraic approach to perturbative quantum field theory, we study the associated, Wick-ordered, quantum observable. In particular we generalize a construction, first developed in the free field theory scenario by Moretti in \cite{moretti}, aimed at exploiting the existing freedoms in the definition of the classical stress-energy tensor, in order to define a quantum counterpart which is divergence free. We focus on Minkowski spacetime proving that this procedure can be adapted also to cubic or quartic self-interactions, at least up to order $\mathcal{O}(\lambda^3)$ in perturbation theory. We remark that this result can be extended to arbitrary globally hyperbolic spacetimes, although, in this case one needs to exploit the existing regularization freedom in the construction of the Wick ordered stress-energy tensor.
\end{abstract}

\allowdisplaybreaks
\section{Introduction}

In the realm of classical and quantum field theory both on flat and on curved, globally hyperbolic, spacetimes, one of most relevant observables is the stress-energy tensor. On the one hand, it encodes the information of the energy content of the underlying physical system, while, on the other hand, it is the source term for the Einstein equations, hence capturing how matter fields influence the spacetime geometry. At a classical level, among its properties, the most notable one is that, if one considers on-shell field configurations, it is divergence free. This is a feature which encodes energy-momentum conservation and it is essential for Einstein's equations to be well-posed since also the Einstein tensor is per construction divergence free.

When one switches from the classical to the quantum realm, the scenario changes drastically. As a matter of fact, even considering for simplicity a non-interacting real scalar field, the stress-energy tensor is quadratic in the field configuration and therefore its quantum counterpart cannot be defined directly since it would involve ill-defined products of distributions. The way out from this obstruction lies in considering Wick ordered products of the underlying fields. While, on Minkowski spacetime, this procedure is usually implemented in terms of normal ordering of the annihilation and creation operators, on a generic curved background this approach is no longer feasible. 

The algebraic approach to quantum field theory, see {\it e.g.} \cite{brunetti2}, is a framework which is best suited to discuss the quantization of a field theory on a globally hyperbolic spacetime. One can summarize it as a two-step procedure. In the first one, one identifies a suitable $*$-algebra of observables $\mathcal{A}$ which encodes all structural properties, such as covariance, dynamics, locality and the canonical commutation or anti-commutation relations. In the second one, one selects an algebraic state, that is a positive and normalized functional on $\mathcal{A}$, which allows to recover the standard probabilistic interpretation proper of quantum theories via the renown GNS theorem. Among the plethora of possible states, only a handful can be recognized as being physically sensible. They are characterized by the {\em Hadamard property} which is a constraint on the singular structure of the underlying two-point correlation function \cite{khav}. In turn, this guarantees that the quantum fluctuations of all observables are finite and that it is possible to identify a local and covariant algebra of Wick polynomials, see also \cite{hollands, hollands1}. This last feature is particularly noteworthy since it entails the possibility of giving a local and covariant definition of a Wick-ordered, quantum stress energy tensor. Yet, in order to ensure covariance, the procedure of constructing a Wick polynomial consists of the subtraction of the singular, covariant component of the two-point function of a Hadamard state. This is known as Hadamard parametrix and it has two notable properties. On the one hand, it is a bi-solution of the underlying equation of motion only up to a smooth remainder, while, on the other hand, its definition is not unique and this leads to the existence of regularization freedoms in the construction of Wick polynomials \cite{khav}. The first property has notable consequences in the definition of the quantum stress-energy tensor since being divergence free relied at a classical level on choosing on-shell configurations. Therefore this feature is no longer automatically assured when considering the Wick-ordered, quantum stress-energy tensor. Yet, at the same time, consistence with Einstein's equations entails that one cannot give up on it with a light heart. Still considering a free scalar field, this problem has been thoroughly investigated and solved in the literature, for example in \cite{hollands2} using the principle of general local covariance and exploiting the freedoms in the definition of the Hadamard parametrix. 

In this work, we consider instead a complementary approach discussed in \cite{moretti}. The starting point is the observation that, in the definition of the classical stress-energy tensor, there is an underlying freedom. Calling $P_0$ the Klein-Gordon operator and $\phi$ the underlying scalar field, one can always add a contribution of the form $\eta g_{\mu\nu}\phi P_0\phi$, where $\eta\in\mathbb{R}$, while $g$ is the background metric. If one considers at a classical level on-shell configurations, this term is vanishing, conserved and traceless. Yet, when promoting it at a quantum level, these features are lost due to Wick ordering. In \cite{moretti} it is shown that, if and only if one sets $\eta=\frac{1}{3}$, then the Wick-ordered quantum stress-energy tensor is divergence free. The price to pay for this result is that, at the level of trace, an additional, state independent contribution appears. Since this depends only on the mass of the field and on the geometry of the underlying background, this is known as {\em trace (or conformal) anomaly}, see \cite{wald} and see also \cite{Ferrero:2023unz} for a generalized, universal notion of the trace anomaly for theories which are not conformally invariant at the classical level. This feature is not proper only of a scalar field theory but it appears also in the analysis of other models, see {\it e.g.} \cite{Dappiaggi:2009xj,Frob:2019dgf} for the investigation of spinor fields.

As far as free field theories are concerned, these results lead to a satisfactory scenario, but, the situation is less crystal clear if one considers instead interacting models. In this paper we analyze a self-interacting, real and massive scalar field on a four-dimensional globally hyperbolic spacetime and we investigate the structural properties of the associated stress-energy tensor. While, at a classical level nothing changes, the non linear nature of the dynamics alters drastically the behavior of the system at a quantum level. We shall work within a framework known as perturbative algebraic quantum field theory (pAQFT). This has been developed mainly in the past fifteen years, see {\it e.g.} \cite{brunetti2,deutsch,rejzner}, and it represents a mathematical formalization of the perturbative approach to interacting field theories which has the notable property of being directly applicable also to models on curved backgrounds. 

It is worth mentioning that pAQFT has been very successful in making precise several aspects of perturbation theory, especially in connection to Epstein-Glaser renormalization, see also \cite{brunetti1} and in developing an algebraic approach to gauge theories \cite{Fredenhagen:2011an,Fredenhagen:2011mq}. It is also versatile enough to allow proving the convergence of the perturbative series in certain instances such as the sine-Gordon model, see \cite{Bahns:2016zqj,Bahns:2021klc}, but also \cite{Bonicelli:2021uxu, bonicelli} for a recent application of pAQFT to the realm of complex systems. Yet, if one restricts the attention to Minkowski spacetime, pAQFT is still seen by several groups working in theoretical physics as a mathematically complicated framework, which is furthermore intrinsically developed in position space and, therefore, it loses at a computational level several of the advantages brought by Feynman diagrams. We are strongly convinced that this viewpoint does not make justice of pAQFT, in particular of its potentialities in better clarifying some conceptual aspects of quantum field theory. In particular, in this work we will be interested in studying the stress-energy tensor at a quantum level for a real and massive scalar field with a cubic or a quartic self-interaction. Exactly as in the free field scenario, the necessity of replacing classical products of fields with suitable Wick polynomials leads to a loss of the conservation law codified by the stress-energy tensor being divergence free. Using the framework of pAQFT and adapting the approach advocated in \cite{moretti}, we will show that it is still possible to add at a classical level a term which is on-shell vanishing, conserved and traceless. Yet, at a quantum level, it allows us to codify that the quantum stress-energy tensor is divergence free. More precisely we shall prove this statement in Minkowski spacetime up to order $\mathcal{O}(\lambda^3)$ in the coupling constant, although there are no structural hurdles to go at higher order, except for an exponential increase in the computational complexity due to the insurgence of new contributions. The reason for focusing our attention on this background is our desire to make contact with models of interest in hadronic physics as explained below. Yet our approach and all the formulae we obtain can be used on a generic globally hyperbolic spacetime. Yet, in this case, it turns out that, in order to cancel the additional terms in the divergence of the stress-energy tensor due to the interacting potential, we are left with a remaining contribution. This is constructed out of geometric tensors, and, therefore, it is both local and covariant. In Minkowski spacetime, as well as in every maximally symmetric globally hyperbolic background, it is automatically zero, but, in general, one can cancel it by exploiting that the Wick-ordered stress-energy tensor admits a large class of regularization ambiguities, as already discussed in \cite{hollands2}. We feel that this is rather remarkable since, contrary to expectations, it points in the direction that the approach advocated in \cite{moretti} also works for interacting models, provided one takes into account all the ambiguities discussed in \cite{hollands2}. One of the reasons for this unexpected outcome lies in the fact that, after adding an on-shell vanishing, conserved and traceless contribution following \cite{moretti}, the stress-energy tensor becomes quadratic in the underlying fields. However, it is unclear whether this property would survive at higher orders in perturbation theory.

From a physical viewpoint the problem under investigation is of interest since the trace anomaly plays a crucial r\^ole in the computation of a property of the energy momentum tensor (EMT) known as \emph{D-term} -- see \cite{maynard, poly}. Unlike other EMT form factors, the D-term is not framed by fundamental properties of particles, like mass or spin, but, being related to the variation of the spatial components of the spacetime metric, it can be expressed in terms of the stress-energy tensor.  
However, in view of the complexity of hadronic structures and of computations in quantum chromodynamics, it is convenient to study the D-term in the case of simpler theories. For instance, the D-term of a free, scalar field theory on Minkowski spacetime has been computed in the literature -- see \cite{hudson}. Therefore, the next sensible step is to investigate how the D-term associated to the scalar theory changes under an infinitesimally small interaction. This question has been addressed up to one-loop order in \cite{maynard}, where it was shown that the D-term for the $\Phi^4$ theory is strongly affected by interactions. The same finding was obtained for the $\Phi^3$ theory, suggesting that this behavior should be independent from the type of interaction. This discussion entails that the D-term must be the particle property most sensitive to the variations of the dynamics in a system. As a consequence, we expect that it can be retrieved from the trace anomaly of the interacting stress-energy tensor and this represents a key motivation for the realization of this work.


\vskip .2cm

The paper is organized as follows: in Section \ref{Sec: Functional Spaces} we recollect some basic notions concerning the functional spaces which are mainly used in the algebraic approach to perturbative quantum field theories, while in Section \ref{Sec: Hadamard States}, we define Hadamard states. We highlight their microlocal properties and we introduce the notions of Hadamard and of Feynman parametrices, the latter necessary to define a time-ordered product. In Section \ref{Sec: Introduction to pAQFT} we introduce the key ingredients of pAQFT. In Section \ref{Sec: Deformation Quantization} we explain how quantization can be seen as a deformation of a classical algebra of observables. This allows us to discuss interacting theories, see Section \ref{Sec: Interacting Theories in pAQFT}, by means of the S-matrix and of the Bogoliubov map introduced in Section \ref{Sec: Bogoliubov map}. We present an explicit example in Section \ref{Sec: Examples}. In Section \ref{Sec: Free Stress-Energy Tensor} we revert our attention to free field theories and we rewrite the approach of \cite{moretti} in the language of functionals. The main results of this work are discussed instead in Section \ref{Sec: Interacting Stress-Energy Tensor} where we consider a massive, real scalar field with a cubic or a quartic self-interaction. We prove that a generalization of the work of \cite{moretti} to these cases is possible and we show that the Wick-ordered, the quantum stress-energy tensor is divergence free up to order $\mathcal{O}(\lambda^3)$ in perturbation theory. We establish this result under the specific assumption that $v_1(z,z)$ as per Equation \eqref{Eq: v1} is a constant function. This property holds true on maximally symmetric backgrounds, such as Minkowski or de Sitter spacetimes.

In addition we compute for a quartic self-interaction the corrections to the trace anomaly. In the appendix, we survey the key concepts at the heart of microlocal analysis, which represents one of the main mathematical tools that we use in this work.

\subsection{Functional spaces}\label{Sec: Functional Spaces}
This section is devoted to an outline of the preliminary tools necessary for the analysis of a free as well as of an interacting quantum field theory within the formalism of perturbative algebraic quantum field theory -- (p)AQFT. We begin by defining the main objects of such approach, which are \emph{functional spaces}.\\
Let $\mathcal{M}\equiv(\mathcal{M},g)$ be a $d$-dimensional, $d\geq 2$,  globally hyperbolic spacetime and we call {\em configuration space}, the set $\mathcal{E}(\mathcal{M})\equiv C^\infty(\mathcal{M})\equiv C^\infty(\mathcal{M};\mathbb{R})$ endowed with the Fréchet topology induced by the family of seminorms.
\begin{equation*}
    ||f||_{k, C} := \sup \{|f^{(k)}(x)|, x \in C\}. 
\end{equation*}
Adopting a terminology typical of pAQFT, we call \emph{space of functionals} on $\mathcal{E}(\mathcal{M})$, $\mathcal{F}(\mathcal{M})$, the space of smooth functionals from $\mathcal{E}(\mathcal{M})$ to $\mathbb{C}$ in the Bastiani sense -- see \cite[Def. 3.4]{rejzner}. This concept is based mainly on a natural notion of Fréchet derivative, whose definition we recall for the reader's convenience.

\begin{definition}[Functional derivative]
\label{funder}
    Let $\mathcal{M}$ be a globally hyperbolic spacetime and let $\mathcal{F}(\mathcal{M})$ be the space of functionals defined above. Given $F \in \mathcal{F}(\mathcal{M})$, we call \emph{$k-$th order functional derivative of $F$} the functional-valued distribution $F^{(k)} \in \mathcal{E}'(\underbrace{\mathcal{M} \times ... \times \mathcal{M}}_k; \mathcal{F}(\mathcal{M})), k \ge 1$ such that 
    \begin{equation*}
        F^{(k)}(\eta_1 \otimes ... \otimes \eta_k; \phi) := \frac{\partial^{(k)}}{\partial s_1 ... \partial s_k} F(\phi + s_1 \eta_1 + ... + s_k \eta_k) \big \vert_{s_1 = ... = s_k = 0}, 
    \end{equation*}
    for all $\phi, \eta_j \in \mathcal{E}(\mathcal{M})$, $j = 1, ..., k$. Furthermore, we call directional derivative along $\eta \in \mathcal{E}(\mathcal{M})$ as
    \begin{equation}
        \delta_{\eta}: \mathcal{F}(\mathcal{M}) \rightarrow \mathcal{F}(\mathcal{M}), \, \, \, \, [\delta_{\eta} F](\phi) := F^{(1)}(\eta; \phi). 
    \end{equation}
\end{definition}
The structures introduce are tailored to the study of a real scalar field, which is the key objection under investigation in this paper. We remark that it is possible to generalize the whole construction to the case of an arbitrary, finite rank, real or complex vector bundle, but we avoid delving in this topic since it is too far from the scopes of this work.\\

Although the space of functionals that we consider can be endowed with an algebra structure, usually called pointwise product, see Definition \ref{microcaualg} below, one of the pillars of our framework consists of exploiting the possibility of deforming this structure in such a way to encode suitable information on the specific physical model under scrutiny. Yet, this leads to considering the product among the derivatives of the underlying functionals and a suitable class of distributional kernels. Due to H\"ormander's criterion -- see Appendix \ref{First appendix} --, this operation is not always well-defined. Therefore, in order to bypass this hurdle, we need to impose constraints on the singular structure of the functionals. In the following we shall make extensive use of concepts and tools proper of microlocal analysis. A succinct summary is available in Appendix \ref{First appendix}, while an extensive analysis can be found in \cite{horm1}.

\begin{definition}[Microcausal functionals]
\label{microfun}
    Let $\mathcal{M}$ be a globally hyperbolic spacetime and let $\mathcal{F}(\mathcal{M})$ be the space of functionals of a real scalar field thereon. We call \emph{microcausal functionals} the elements of 
    \begin{equation}
        \mathcal{F}_{\mu c}(\mathcal{M}) := \left \{ F \in \mathcal{F}(\mathcal{M}) \, \vert \, F^{(k)} \in \mathcal{E}'(\mathcal{M}^k), WF(F^{(k)}) \subset G_k (\mathcal{M}), \forall k \in \mathbb{N} \right\},
    \end{equation}
    where $F^{(k)}$ denotes the $k-$th functional derivative of $F$ as per Definition \ref{funder} while 
    \begin{equation*}
        G_k(\mathcal{M}) := T^*(\underbrace{\mathcal{M} \times ... \times \mathcal{M}}_k \setminus \left( \bigcup_{x \in \mathcal{M}} (V^+_x)^k\cup \bigcup_{x \in \mathcal{M}} (V^-_x)^k \right) ), 
    \end{equation*} 
    $(V^{\pm}_x)^k$ being the counterpart in $T^*_x (\underbrace{\mathcal{M} \times ... \times \mathcal{M}}_k)$ of the future ($+$) and past ($-$) light cones at $x \in \mathcal{M}$. In addition we denote by $\mathcal{P}_{\mu c}(\mathcal{M})$ the collection of all {\em polynomial} functionals, namely $F\in\mathcal{P}_{\mu c}(\mathcal{M})$ if there exists $\overline{k}\in\mathbb{N}$ such that $F^{(k)}=0$ for all $k\geq\overline{k}$. 
\end{definition}
\noindent Among these functionals, two notable subclasses can be identified. 
\begin{definition}[Regular functionals]
     Let $\mathcal{M}$ be a globally hyperbolic spacetime and let $\mathcal{F}_{\mu c}(\mathcal{M})$ be as per Definition \ref{microfun}. A \emph{regular functional} is an element of 
     \begin{equation}
         \mathcal{F}_{reg}(\mathcal{M}) := \left\{F \in \mathcal{F}_{\mu c}(\mathcal{M}) \; \vert \; F^{(k)} \in \mathcal{D}(\mathcal{M}^k) \hookrightarrow \mathcal{E}'(\mathcal{M}^k), \forall k \in \mathbb{N} \right\}.
     \end{equation}
\end{definition}

\begin{definition}[Local functionals]\label{Def: local functionals}
     Let $\mathcal{M}$ be a globally hyperbolic spacetime and let $\mathcal{F}_{\mu c}(\mathcal{M})$ be the space of microcausal functionals as per Definition \ref{microfun}. Let us denote the spacetime support of a functional $F \in \mathcal{F}_{\mu c}(\mathcal{M})$ by
     \begin{equation}\label{Eq: support of a functional}
         supp(F) :=\{ x \in \mathcal{M} \, \vert \, \forall U \in \mathcal{N}_x, \exists \eta, \theta \in \mathcal{E}(\mathcal{M}) : supp(\eta) \subset U, F(\eta + \theta) \ne F({\theta})\}, 
     \end{equation}
     $\mathcal{N}_x$ being the set of all open neighborhoods of $x \in \mathcal{M}$. We call \emph{local functional} an element of 
     \begin{equation}\label{Eq: local functionals}
         \mathcal{F}_{loc}(\mathcal{M}) := \left\{F \in \mathcal{F}_{\mu c}(\mathcal{M}) \, \vert \, supp(F^{(k)}) \subset Diag_k(\mathcal{M}), \forall k \in \mathbb{N} \right\},
     \end{equation}
     where 
     \begin{equation*}
         Diag_k(\mathcal{M}) := \left\{ \underbrace{(x, ..., x)}_k \in \mathcal{M}^k \, \vert \, x \in \mathcal{M} \right\}.
     \end{equation*}
    Similarly to Definition \ref{microfun}, we denote by $\mathcal{P}_{loc}(\mathcal{M})$ the collection of all local and polynomial microcausal functionals. 
\end{definition}
\noindent Let us discuss an informative example to illustrate the definitions given so far. 
\begin{example}
    Let $\mathcal{M}$ be a globally hyperbolic spacetime and choose $f \in \mathcal{D}(\mathcal{M})$. We can construct the \emph{smeared real scalar field} as
    \begin{flalign}
        \Phi_f &: \mathcal{E}(\mathcal{M}) \rightarrow \mathbb{C}, \notag\\
        \phi &\mapsto \Phi_f(\phi) := \int_{\mathcal{M}} d\mu_x\, \phi(x) f(x),\label{Eq: Field Generator}
    \end{flalign}
    where $d\mu_x$ is the metric-induced volume measure. A direct computation of the functional derivatives of $\Phi_f$ yields for all $\phi,\eta,\eta_1,\dots\eta_k\in\mathcal{E}(M)$
    \begin{flalign*}
        \Phi_f^{(1)} (\eta; \phi) &:= [\delta_{\eta} \Phi_f](\phi) = \int_{\mathcal{M}^2}d\mu_x \, d\mu_y\, \delta_2(x,y) f(x) \eta(y), \\
        \Phi_f^{(k)} (\eta_1 \otimes ... \otimes \eta_k; \phi) &= 0, \, \forall k \ge 2,
    \end{flalign*}
    where $\delta_2(x,y)$ is the formal integral kernel of the Dirac delta $\delta_2 \in \mathcal{D}'(\mathcal{M} \times \mathcal{M})$. Since $supp(\delta_2) \subset Diag_2(\mathcal{M})$, we can infer that $\Phi_f \in \mathcal{F}_{loc}(\mathcal{M}), \forall f \in \mathcal{D}(\mathcal{M})$.
\end{example}

\begin{remark}\label{Rem: Derivatives applied to a functional}
It is possible to generalize Definition \ref{Def: local functionals} as highlighted in \cite[Def. 3.14]{rejzner}. More precisely one could call local any functional $F$ such that, for every $\phi_0\in\mathcal{E}(\mathcal{M})$ there exists an open neighborhood $U_{\phi_0}\in\mathcal{E}(\mathcal{M})$ as well as $k\in\mathbb{N}$ such that for all $\phi\in U_{\phi_0}$
$$F(\phi)=\int\limits_{M}\alpha_{j^k_x(\phi)},$$
where $j^k_x(\phi)$ is the k-th jet prolongation of $\phi$ while $\alpha$ is a density-valued function on the underlying jet bundle. This allows to define precisely functionals which appear in the definition of quantities of physical interest such as the stress-energy tensor. As an example consider 
$$(\partial_\mu\Phi\partial_\nu\Phi)[\phi,h]=\int\limits_{\mathcal{M}}d\mu_x\,\partial_\mu\phi(x)\partial_\nu\phi(x)h^{\mu\nu}(x),$$
where $h\in\Gamma_0(TM^{\otimes_s 2})$. It is worth stressing that, with a slight abuse of notation, we do not suppress the dependence on the indices in the above equation, even though the functional is tested against a test function. In addition, since denoting every single time the test function becomes a burden at the level of notation, we shall privilege working at the level of integral kernels.
\end{remark}

\noindent Henceforth, we will work mainly with polynomial functionals since they suffice to analyze the models of interest for this paper. Starting from Definition \ref{microfun} and thanks to Theorem \ref{Thm: Hormander criterion}, we can introduce a *-algebra structure as follows. 
\begin{definition}[Classical Microcausal Algebra]
\label{microcaualg}
    Let $\mathcal{M}$ be a globally hyperbolic spacetime. We call \emph{classical microcausal algebra} the commutative, unital, *-algebra $\mathcal{A}_{cl}(\mathcal{M}) := (\mathcal{F}_{\mu c}(\mathcal{M}), \cdot, \ast)$, where 
    \begin{itemize}
        \item $\mathcal{F}_{\mu c}(\mathcal{M})$ is the space of microcausal functionals associated to a real scalar field as per Definition \ref{microfun}, 
        \item the composition map 
        \begin{flalign*}
            \cdot: &\mathcal{F}_{\mu c}(\mathcal{M}) \times \mathcal{F}_{\mu c}(\mathcal{M}) \rightarrow \mathcal{F}_{\mu c}(\mathcal{M}) \\
            (F,G) & \mapsto F \cdot G := \iota^*(F \otimes G),
        \end{flalign*}
        $\iota^*$ is the pullback on $\mathcal{F}_{\mu c}(\mathcal{M}) \times \mathcal{F}_{\mu c}(\mathcal{M}) $ along the diagonal map 
        \begin{flalign*}
           \iota:  \, &\mathcal{E}(\mathcal{M}) \rightarrow \mathcal{E}(\mathcal{M}) \times \mathcal{E}(\mathcal{M})\\
           \phi & \mapsto \iota(\phi) := (\phi, \phi). 
        \end{flalign*}
        The map $\cdot$ is commutative and it is called pointwise product, \textit{i.e.}, $\forall F,G \in \mathcal{F}_{\mu c}(\mathcal{M})$, $\forall \phi \in \mathcal{E}(\mathcal{M})$, 
        \begin{equation}\label{Eq: Pointwise product}
            (F \cdot G)(\phi) = \iota^{\ast} \left[F\otimes G\right] (\phi),
        \end{equation}
        \item the involution is given by 
        \begin{flalign*}
            \ast: \mathcal{F}_{\mu c}(\mathcal{M}) &\rightarrow  \mathcal{F}_{\mu c}(\mathcal{M}) \\
            F^*(\phi) &:= \overline{F(\phi)},
        \end{flalign*}
        where $F \in \mathcal{F}_{\mu c}(\mathcal{M})$, $\phi \in \mathcal{E}(\mathcal{M})$.
    \end{itemize}
    In addition we denote by $\mathrm{Pol}_{cl}(\mathcal{M})$ the $*$-subalgebra of $\mathcal{A}_{cl}(\mathcal{M})$ built out polynomial, microcausal functionals.
\end{definition}

\noindent We observe that Equation \eqref{Eq: Pointwise product} might be ill-defined since it involves the pull-back of the tensor product of two distributions along the diagonal. Yet, all the assumptions that we make on the structural properties of the underlying functionals entail that this is not case, as one can see using microlocal techniques \cite{horm1}.

\begin{remark}\label{Rem: Generators}
    For later convenience, observe that the algebra of polynomial microcausal functionals includes
    \begin{enumerate}
        \item the identity functionals $\textbf{I}_f$ such that 
    $$\textbf{I}_f(\phi)\doteq\int\limits_{\mathcal{M}}d\mu_x\, f(x),$$
    where $d\mu_x$ is the metric induced volume form.
    \item the powers of the smeared field, namely, for $k\in\mathbb{N}$ and $f\in\mathcal{D}(\mathcal{M})$
     \begin{flalign*}
        \Phi^k_f &: \mathcal{E}(\mathcal{M}) \rightarrow \mathbb{C}, \\
        \phi &\mapsto \Phi^k_f(\phi) := \int_{\mathcal{M}}d\mu_x\, \phi^k(x) f(x).
    \end{flalign*}
    \end{enumerate}
\end{remark}

\subsection{Hadamard states}\label{Sec: Hadamard States}

In this section we review the second key ingredient used in the formulation of pAQFT, namely quantum states. We recall that, given a unital $*$-algebra $\mathcal{A}$, an (algebraic) state is a linear, normalized and positive functional $\omega:\mathcal{A}\to\mathbb{C}$, namely
$$\omega(\mathbb{I})=1,\quad\omega(a^*a)\geq 0,\;\forall a\in\mathcal{A},$$
where $\mathbb{I}$ is the identity element lying in $\mathcal{A}$. On the one hand, the pair $(\mathcal{A},\omega)$ allows to recover the probabilistic interpretation proper of a quantum theory by means of the GNS theorem. On the other hand, among the plethora of existing states, not all are physically sensible and this conundrum can be avoided by imposing the Hadamard condition. We review its definition and consequences with a particular attention to the structures relevant in pAQFT. For definitiveness we restrict our attention to the scenario of interest, namely a self-interacting, real scalar field.  

In particular, the free dynamics is ruled by the Klein-Gordon operator
\begin{equation}\label{Eq: KG}
P_0 =  -\Box_g + m_{\phi}^2+\xi R,
\end{equation}
where $\Box_g$ is the d'Alembert wave operator constructed out of the underlying Lorentzian metric assumed to have signature $(-,+,\dots,+)$ and locally given by $\nabla^{\mu} \nabla_{\mu}$, $m_{\phi} \geq 0$ is a mass parameter, $\xi\in\mathbb{R}$, while $R$ is the scalar curvature built out of $g$. Since the background $\mathcal{M}$ is a globally hyperbolic spacetime, associated to $P_0$, there exist unique advanced and retarded fundamental solutions $\Delta_{A/R}\in\mathcal{D}^\prime(\mathcal{M}\times \mathcal{M})$ such that $P_0\circ\Delta_{A/R}=\Delta_{A/R}\circ P_0=\delta_{\mathrm{Diag}_2}$ and $\textrm{supp}(\Delta_{A/R}(f))\subseteq J^\mp(\textrm{supp}(f))$ for all $f\in\mathcal{D}(\mathcal{M})$ \cite{bar}. Here $J^\mp$ denote the causal future ($+$) and past ($-$) while $\delta_{\mathrm{Diag}_2}$ is the Dirac delta supported on the diagonal of $\mathcal{M}\times\mathcal{M}$. Associated to $\Delta_{A/R}$ we can define two notable distributions, the {\em causal (Pauli-Jordan) propagator}
\begin{equation}\label{Eq: Causal Propagator}
    \Delta\doteq\Delta^R-\Delta^A,
\end{equation}
and the {\bf Dirac propagator}
\begin{equation}\label{Eq: Dirac Propagator}
\Delta_D=\frac{1}{2}\left(\Delta_A+\Delta_R\right).
\end{equation}
Observe that both $\Delta$ and $\Delta_D$ inherit from the $\Delta_{A/R}$ the property of being local and covariant. With these ingredients and recalling the characterization of wavefront set reported in Appendix \ref{First appendix}, we can give the following definition.

\begin{definition}\label{Def: Hadamard 2-pt}
Given a $d$-dimensional, $d\geq 2$, globally hyperbolic spacetime $\mathcal{M}$, we say that $\Delta_+\in\mathcal{D}^\prime(\mathcal{M}\times\mathcal{M})$ is a {\bf (global) Hadamard two-point function} if 
\begin{enumerate}
    \item $\left(P_0\otimes\mathbb{I}\right)\Delta_+=\left(\mathbb{I}\otimes P_0\right)\Delta_+=0$, 
    \item for any pair of test-functions $f,f^\prime\in\mathcal{D}(\mathcal{M})$
    $$\Delta_+(\overline{f},f)\geq 0\;\textrm{and}\;\Delta_+(f,f^\prime)-\Delta_+(f^\prime,f)=i\Delta(f,f^\prime),$$
    where $\Delta\doteq\Delta^R-\Delta^A$ is the causal (Pauli-Jordan) propagator,
    \item for any pair of test-functions $f,f^\prime\in\mathcal{D}(\mathcal{M})$
    $$|\Delta(\overline{f},f^\prime)|^2\leq 4\Delta_+(\overline{f},f)\Delta_+(\overline{f}^\prime,f^\prime),$$
    \item the wavefront set of $\Delta_+$ is
    \begin{equation}\label{Eq: WF Delta+}
        \mathrm{WF}(\Delta_+)=\{(x,k_x,y,-k_y)\in T^*(\mathcal{M}\times\mathcal{M})\setminus\{0\}\;|\;(x,k_x)\sim (y,k_y)\;\textrm{and}\; k_x\triangleright 0\}
    \end{equation}
\end{enumerate}
where the symbol $\sim$ implies that $x$ and $y$ are connected by a lightlike geodesic $\gamma$ to which $k_x$ is cotangent, while $k_y$ is obtained by parallel transport along $\gamma$. Furthermore $k_x\triangleright 0$ entails that $k_x$ is future-pointing.
\end{definition}

\begin{remark}
    It is worth emphasizing that, in Definition \ref{Def: Hadamard 2-pt}, only condition $4$ establishes a constraint on the singular structure of the underlying two-point function and, in some instances, it is considered as the lone defining condition for a Hadamard two-point function. Condition 1-3 are, instead, necessary to guarantee that $\Delta_+$ identifies unambiguously a quasi-free/Gaussian state on the algebra of observables for a free scalar field theory, whose dynamics is ruled by $P_0$, see for example \cite[Chap. 5]{brunetti2}. Although we shall not enter in the details of the construction of such algebra, we will always require condition 1-3 to hold true, since this is a necessary prerequisite for the consistency of the analysis, discussed in the following sections.
\end{remark}

Working with Hadamard states is necessary for several reasons \cite{Fewster:2013lqa}, among which noteworthy are the finiteness of the quantum fluctuations of all observables and the possibility of constructing a local and covariant algebra of Wick polynomials, see \cite{hollands1, hollands2} and \cite[Chap. 5]{brunetti2}. In particular, this last result is based on the following local characterization of Hadamard states.

\begin{definition}\label{Def: local Hadamard 2-pt}
Given an $d$-dimensional, $d\geq 2$, globally hyperbolic spacetime $\mathcal{M}$, we say that $\Delta_+\in\mathcal{D}^\prime(\mathcal{M}\times\mathcal{M})$ is a {\bf (local) Hadamard two-point function} if it abides by the condition 1-3 in Definition \ref{Def: Hadamard 2-pt} and if, working at the level of integral kernels, for all geodesically convex neighborhood $\mathcal{O}\subseteq\mathcal{M}$ and for all $x,y\in\mathcal{O}$
\begin{equation}\label{Eq: Hadamard two-point function}
\Delta_+(x,y)=\lim\limits_{\epsilon\to 0^+}\frac{U_d(x,y)}{\sigma^{\frac{d-2}{2}}_\epsilon(x,y)}+V_d(x,y)\ln \frac{\sigma(x,y)}{\zeta_H}+W(x,y),
\end{equation}
where $\sigma(x,y)$ is the halved, squared geodesic distance between $x$ and $y$ while, denoting by $t:\mathcal{M}\to\mathbb{R}$ a time function, $\sigma_\epsilon(x,y)\doteq\sigma(x,y)+i\epsilon(t(x)-t(y))+\epsilon^2$. In addition $\zeta_H \in\mathbb{R}^+$ is a reference squared length, while $U_d,V_d,W\in C^\infty(\mathcal{O}\times\mathcal{O})$. The singular terms in Equation \eqref{Eq: Hadamard two-point function} identify the {\bf Hadamard parametrix} $H\in\mathcal{D}^\prime(\mathcal{O}\times\mathcal{O})$, whose integral kernel is such that
\begin{equation}\label{Eq: Hadamard parametrix}
    H(x,y)=\lim\limits_{\epsilon\to 0^+}\frac{U_d(x,y)}{\sigma^{\frac{d-2}{2}}_\epsilon(x,y)}+V_d(x,y)\ln \frac{\sigma(x,y)}{\zeta_H}.
\end{equation}
\end{definition}

The Hadamard parametrix depends critically on $\zeta_H$. This plays a distinguished r\^{o}le in individuating the freedom in the local and covariant construction of Wick polynomials, but, with a slight abuse of notation, we shall leave this dependence implicit to avoid a heavy notation. As a matter of fact such freedom will not play a crucial r\^{o}le in the derivation of our results and therefore, highlighting it when unnecessary, might jeopardize a smooth reading of this work.

It is most notable that the smooth functions $U_d$ and $V_d$ can be constructed from the free dynamics codified in condition $1$ in Definition \ref{Def: Hadamard 2-pt}. More precisely they can be both expanded in power series with respect to the geodesic distance. This is known as the {\em Hadamard expansion}: Splitting between even and odd dimensions,
\begin{subequations}\label{Eq: Hadamard expansion}
\begin{equation}\label{Eq: Hadamard expansion a}
U_{2n}(x,y)=\Theta_{2n}\sum\limits_{k=0}^{n-2}u_k(x,y)\sigma^k(x,y)\quad \textrm{and}\quad V_{2n}(x,y)=\sum\limits_{k=0}^\infty v_k(x,y)\sigma^k(x,y),
\end{equation}
\begin{equation}\label{Eq: Hadamard expansion b}
U_{2n+1}(x,y)=\sum\limits_{k=0}^{\infty}u_k(x,y)\sigma^k(x,y)\quad \textrm{and}\quad V_{2n+1}(x,y)=0,
\end{equation}
\end{subequations}
where $\Theta_2=0$ and $\Theta_{2n}=1$ if $n>1$. The coefficients $u_k$ and $v_k$ can be determined by solving a family of transport equations, known as {\em Hadamard recursion relations}. Yet, unless the spacetime is analytic, the series in Equation \eqref{Eq: Hadamard expansion} is only asymptotic \cite{Friedlander:2010eqa}. We shall not delve into the details of this construction, leaving an interested reader to \cite{moretti}. We content ourselves with remarking that, in view of this result, the Hadamard parametrix $H$ codifies the universal, {\em local and covariant component} of any Hadamard state, while the function $W$ in Equation \eqref{Eq: Hadamard two-point function} encompasses our freedom in the choice of a state falling in this class. To conclude this succinct excursus, we report a renown result first proven by Radzikowski in \cite{radzi, radzi2}. It shows the equivalence between the local and the global form of a Hadamard state.

\begin{theorem}[Radzikowsky]\label{Thm: Radzikowski}
    Let $\mathcal{M}$ be a globally hyperbolic spacetime and let $\Delta_+\in\mathcal{D}^\prime(\mathcal{M}\times\mathcal{M})$. It abides by Definition \ref{Def: Hadamard 2-pt} if and only if it also satisfies the hypotheses of Definition \ref{Def: local Hadamard 2-pt}.
\end{theorem}

In addition to Hadamard states, another primary building block of our framework is the {\em Feynman propagator}
\begin{equation}\label{Eq: Feynman propagator}
\Delta_F=\Delta_++i\Delta_A\in\mathcal{D}^\prime(\mathcal{M}\times\mathcal{M}),
\end{equation}
where $\Delta_+$ is a Hadamard two-point function as per Definition \ref{Def: Hadamard 2-pt}. Observe that, on the one hand, we call $\Delta_F$ a propagator to keep the convention used in the literature although $P_0\Delta_F=i\delta_{\mathrm{Diag}_2}$, while, on the other hand, contrary to the Dirac counterpart in Equation \eqref{Eq: Dirac Propagator}, $\Delta_F$ is non unique since it depends on the choice of $\Delta_+$. Using the theorem of propagation of singularities, one can show \cite{radzi, radzi2} that, denoting $\mathring{T}^*(\mathcal{M}\times\mathcal{M})\doteq T^*(\mathcal{M}\times\mathcal{M})\setminus\{0\}$,
\begin{equation}\label{Eq: WF of HF}
\mathrm{WF}(\Delta_F)=\{(x,k_x,y,-k_y)\in\mathring{T}^*(\mathcal{M}\times\mathcal{M}) \;|\;(x,k_x)\sim_F (y,k_y)\}\cup\mathrm{WF}(\delta_{\mathrm{Diag}_2}),
\end{equation}
where $\mathrm{WF}(\delta_{\mathrm{Diag}_2})=\{(x,k_x,x,-k_x)\in\mathring{T}^*(\mathcal{M}\times\mathcal{M})\}$, while $\sim_F$ entails that there exists a lightlike geodesic $\gamma$, connecting $x$ to $y$, such that $k_y$
is obtained as the parallel transport of $k_x$ along $\gamma$. In addition $k_x$ is co-tangent to $\gamma$ at $x$ and if $y\in J^-(x)$
then $k_x$ is past directed. Conversely, if $y\in J^+(x)$, then $k_y$ is future pointing.

To avoid this arbitrariness related to the choice of $\Delta_+$, one can work with the {\bf Feynman parametrix} $H_F\in\mathcal{D}^\prime(\mathcal{M}\times\mathcal{M})$, such that $\Delta_F-H_F\in C^\infty(\mathcal{M}\times\mathcal{M})$. Its existence is guaranteed by Duistermaat-H\"ormander theory of distinguished parametrices \cite{duist}. Similarly to Equation \eqref{Eq: Hadamard parametrix}, we can characterize explicitly the integral kernel of $H_F$, namely, for any pair of points $x,y$ lying in a convex geodesic neighborhood $\mathcal{O}$,
\begin{equation}\label{Eq: Feynman parametrix}
    H_F(x,y)=H+i\Delta_A=\lim\limits_{\epsilon\to 0^+}\frac{U_d(x,y)}{\sigma^{\frac{d-2}{2}}_{F,\epsilon}(x,y)}+V_d(x,y)\ln\frac{\sigma_{F,\epsilon}(x,y)}{\zeta},
\end{equation}
where $\sigma_{F,\epsilon}(x,y)\doteq\sigma(x,y)+i\epsilon$. Here the coefficients $U_d$ and $V_d$ are the same as in Equation \eqref{Eq: Hadamard parametrix}.

\begin{remark}\label{Rem: Powers of HF}
    In the context of perturbative algebraic quantum field theory, we shall consider also powers of the Feynman parametrix $H_F$. Yet a close inspection of Equation \eqref{Eq: WF of HF} shows that H\"ormander criterion for the multiplication of two distributions in not abided by, see Theorem \ref{Thm: Hormander criterion}. Hence, regardless of the spacetime dimension, we can only conclude that, for all $k\geq 2$ $H_F^k\in\mathcal{D}^\prime(\mathcal{M}\times\mathcal{M}\setminus X)$ where $X$ is the collection of points $(x,y)\in\mathcal{M}\times\mathcal{M}$ connected by a lightlike geodesic. 

    This prompts the question whether $H_F^k$ admits an extension to the whole manifold $\mathcal{M}\times\mathcal{M}$. Using Theorem \ref{Thm: Degree of divergence}, it turns out that an answer can be found estimating the scaling degree of $H_F$, see Definition \ref{scaldeg} and \ref{Def: Microlocal scaling degree}. To this end, it is convenient and sufficient to consider Equation \eqref{Eq: Feynman parametrix} according to which the leading singularity is ruled by $\sigma^{\frac{2-d}{2}}$, a power of the geodesic distance, provided that $\dim\mathcal{M}=d>2$. If $d=2$ it turns out that the only singularity occurring in $H_F$ is logarithmic in $\sigma$. Hence, in this case, all powers of $H_F$ are automatically well-defined on the whole $\mathcal{M}\times\mathcal{M}$. Focusing instead on $d>2$, a direct computation yields that, for $k\geq 1$, the scaling degree $\mathrm{sd}(H^k_F)=\mathrm{sd}(\sigma^{\frac{k(2-d)}{2}})=(d-2)k$. Using the language of Theorem \ref{Thm: Degree of divergence}, this entails that the degree of divergence $\rho_{k,F}$ of $H^k_F$ is always finite and therefore one or more extensions to the whole manifold $\mathcal{M}\times\mathcal{M}$ exist. In the physically relevant scenario where $d=4$, it turns out that the $\rho_{k,F}=2k-4$ which is always greater or equal to zero if $k\geq 2$. Hence there is a freedom in extending the powers of the Feynman propagator to $\mathcal{M}\times\mathcal{M}$, namely, as well known, a renormalization procedure is necessary.
\end{remark}

\section{Introduction to perturbative Algebraic Quantum Field Theory}\label{Sec: Introduction to pAQFT}

In this section we give a succinct introduction to perturbative algebraic quantum field theory. Although the framework has a much wider range of applicability, for definitiveness, we shall always have in mind a self-interacting, real, scalar field $\phi:\mathcal{M}\to\mathbb{R}$ with $\dim\mathcal{M}=4$ and whose dynamics is ruled by 
\begin{equation}\label{Eq: non-linear dynamics}
P_0\phi=\frac{\lambda}{(n-1)!}\phi^{n-1},
\end{equation}
where $n = 3, 4$, $\lambda\in\mathbb{R}$ is a coupling constant, while $P_0$ is as per Equation \eqref{Eq: KG}.

\subsection{Deformation quantization}\label{Sec: Deformation Quantization} 

\noindent In the following we outline a deformation quantization scheme which is the procedure at the heart of pAQFT to encode information both about the underlying dynamics and about the building blocks of the quantum theory. This scheme consists of a modification of the product of the algebra $\mathcal{A}_{cl}(\mathcal{M})$ as in Definition \ref{microcaualg} by means of a deformation parameter $\hbar$ and of a bi-distribution $K \in \mathcal{D}^\prime(\mathcal{M} \times \mathcal{M})$, while the underlying space of microcausal functionals is left untouched. As a result we obtain a new algebra $\mathcal{A}_K(\mathcal{M}) := (\mathcal{F}_{\mu c} (\mathcal{M}), \star_{K}, \ast)$, where $\forall F,G \in \mathcal{F}_{\mu c}(\mathcal{M})$, the deformed product $\star_K$ is defined as follows  
\begin{flalign}
\label{starprod}
    F &\star_{K} G = \iota^* \circ e^{D_{\hbar K}} [F \otimes G], \\ \notag 
    D_{\hbar K} &:= \langle \hbar K, \frac{\delta}{\delta \phi} \otimes \frac{\delta}{\delta \phi} \rangle := \int_{\mathcal{M}^2} d\mu_x \, d\mu_y \, \hbar K(x,y) \frac{\delta}{\delta \phi(x)} \otimes \frac{\delta}{\delta \phi(y)},
\end{flalign}
where $K(x,y)$ denotes the integral kernel of the bi-distribution $K$. Let us stress that the exponential map of the differential operator $D_{\hbar K}$ has to be intended as a formal power series in the parameter $\hbar$. In addition each term on the exponential series might involve a priori ill-defined operations among distribution. In order to avoid this last hurdle a specific and physically motivated choice for the bi-distribution $K$ has to be made. Furthermore, to avoid unnecessary technical hurdles, we shall consider only polynomial functionals.
\begin{definition}
\label{quanmicroal}
    Let $\mathcal{M}$ be a globally hyperbolic spacetime and let $\mathrm{Pol}_{cl}(\mathcal{M})$ be the classical microcausal algebra of polynomial functionals associated to a real scalar field theory as per Definition \ref{microcaualg}. Let $\Delta_+ \in \mathcal{D}^\prime(\mathcal{M} \times \mathcal{M})$ be a Hadamard bi-distribution as per Definition \ref{Def: Hadamard 2-pt}. We call \emph{quantum algebra of microcausal functionals} $\mathrm{Pol}_{\Delta_+} (\mathcal{M}) = (\mathrm{Pol}_{\mu c}(\mathcal{M}), \star_{\Delta_+}, \ast)$, whose associative product is obtained via Equation \eqref{starprod}, setting $K = \Delta_+$. 
\end{definition}

\begin{remark}
    Observe that, on account of Equation \eqref{Eq: WF Delta+} and of Definition \ref{microfun}, each term of the exponential series in Equation \eqref{starprod} is well-defined from a microlocal viewpoint, see Appendix \ref{First appendix}. Furthermore, since we consider only polynomial functionals, the action of the exponential series yields only a finite number of non vanishing contributions and hence convergence is not an issue.
\end{remark}

\begin{remark}\label{Rem: Expectation value}
Since in the algebra $\mathrm{Pol}_{\Delta_+} (\mathcal{M})$ the information concerning the choice of an underlying quantum state has already been encoded in the product $\star_{\Delta_+}$, the expectation value of a functional $A\in\mathrm{Pol}_{\Delta_+} (\mathcal{M})$ is codified by the evaluation at the vanishing configuration, namely
\begin{equation*}
    \langle A \rangle = A \vert_{\phi = 0}.
\end{equation*}
\end{remark}

\begin{example}\label{Exam: Wick polynomials}
    Let $\mathcal{M}$ be a globally hyperbolic spacetime and let $\mathrm{Pol}_{\Delta_+}(\mathcal{M})$ be the quantum algebra of microcausal, polynomial functionals associated to a real scalar field thereon as per Definition \ref{quanmicroal}. Recalling Remark \ref{Rem: Generators}, let us consider the local functionals,  
    \begin{equation}
        \Phi^2_f(\phi) := \int_{\mathcal{M}} d\mu_x \, \phi^2(x) f(x), \, \,  \Phi^2_{f'}(\phi) := \int_{\mathcal{M}} d\mu_y \, \phi^2(y) f'(y).\quad \phi \in \mathcal{E}(\mathcal{M}),\;\textrm{and}\; f, f' \in \mathcal{D}(\mathcal{M}) 
    \end{equation}
    A direct computation yields 
    \begin{flalign*}
        \Phi^2_f(\phi) \star_{\Delta_+} \Phi^2_{f'}(\phi) & = \int_{\mathcal{M}} d\mu_x \, \phi^2(x) f(x) \int_{\mathcal{M}} d\mu_y \, \phi^2(y) f'(y) \\ &+ 4\hbar \int_{\mathcal{M}^2} d\mu_x \, d\mu_y \phi(x) \phi(y) \Delta_+(x,y) f(x) f'(y) \\ &+ 2\hbar^2 \int_{\mathcal{M}^2} d\mu_x \, d\mu_y \Delta^2_+(x,y) f(x) f'(y). 
    \end{flalign*}
Observe that $\Delta^2_+$ is a well-defined element of $\mathcal{D}^\prime(\mathcal{M}\times\mathcal{M})$ on account of Equation \eqref{Eq: WF Delta+} and of H\"ormander criterion for the product of two distributions, see Theorem \ref{Thm: Hormander criterion}. 
\end{example}

\noindent
From the viewpoint of microlocal analysis, in order to make sense of Definition \ref{quanmicroal} we are allowed to consider in place of $\Delta_+$ any bi-distribution $H\in\mathcal{D}^\prime(\mathcal{M}\times\mathcal{M})$ such that $\Delta_+-H\in C^\infty(\mathcal{M}\times\mathcal{M})$. One might protest that we are using the same symbol as in Equation \eqref{Eq: Hadamard parametrix} and this is, indeed, a slight abuse of notation motivated by the following remark.

\begin{remark}
    As one can infer from Definition \ref{Def: local Hadamard 2-pt} and from Theorem \ref{Thm: Radzikowski}, if one seeks a local and covariant structure, when restricting the attention to geodesically convex neighborhoods, it is more appropriate to work with the Hadamard parametrix $H\in\mathcal{D}^\prime(\mathcal{O}\times\mathcal{O})$. This allows for example the identification of a local and covariant notion of Wick ordering, see \cite{hollands}. Yet, as pointed out in Definition \eqref{Def: local Hadamard 2-pt}, the restriction to $\mathcal{O}\times\mathcal{O}$ of $\Delta_+$ and therefore also of any $H\in\mathcal{D}^\prime(\mathcal{M}\times\mathcal{M})$ such that $\Delta_+-H\in C^\infty(\mathcal{M}\times\mathcal{M})$, coincides with Equation \eqref{Eq: Hadamard parametrix} up to a smooth remainder. 
\end{remark}

\begin{definition}[Local covariant algebra of functionals]\label{Def: Local and Covariant algebra}
Given a globally hyperbolic spacetime $\mathcal{M}$, we call \emph{local and covariant algebra of microcausal polynomial functionals} $\mathrm{Pol}(\mathcal{M})$, the algebraic closure of $\mathrm{Pol}_{loc}(\mathcal{M})$ as per Definition \ref{Def: local functionals} under the product $\star_H$. This is the deformed product as per Equation \eqref{starprod} with respect to the $H\in\mathcal{D}^\prime(\mathcal{M}\times\mathcal{M})$ such that $\Delta_+-H\in C^\infty(\mathcal{M}\times\mathcal{M})$.
\end{definition}

\begin{definition}\label{Def: Wick ordering}
    Let $\mathrm{Pol}(\mathcal{M})$ be the polynomial, microcausal functionals as per Definition \ref{Def: Local and Covariant algebra}. We call {\em (abstract) Wick-ordering} of $A\in\mathrm{Pol}(\mathcal{M})$
    \begin{equation}\label{Eq: Wick ordering}
    :A:_H := \alpha^{-1}_H(A)=\alpha_{-H}(A) := e^{- \frac{1}{2}\langle H, \frac{\delta^2}{\delta \phi(x)\delta \phi(y)}\rangle}  A, 
    \end{equation}
    where $H$ is chosen as in Definition \ref{Def: Local and Covariant algebra}. 
\end{definition}

\noindent Observe that $\mathrm{Pol}(\mathcal{M})$ is $*$-isomorphic to $\mathrm{Pol}_{\Delta_+}(\mathcal{M})$ via the map $\alpha_W:\mathrm{Pol}_{\mu c} (\mathcal{M})\to\mathrm{Pol}_{\mu c} (\mathcal{M})$ with $W\doteq \Delta_+-H\in C^\infty(\mathcal{M}\times\mathcal{M})$, see \cite{rejzner}. In other words, for all $F,G\in\mathrm{Pol} (\mathcal{M})$
$$F\star_H G=\alpha_W\left(\alpha^{-1}_W(F)\star_{\Delta_+}\alpha^{-1}_W(G)\right),$$
where $\alpha_W$ is defined as in Equation \eqref{Eq: Wick ordering} with $-H$ replaced by $W$.
We have all the ingredients to give a definition of Wick-ordering. Once more, for simplicity, we focus on polynomial functionals. We proceed in two steps.

\noindent We stress once more that, if we restrict the attention to an arbitrary but fixed geodesically convex open neighborhood, we recover the results of \cite{hollands}. This novel algebraic structure allows to interpret the elements of $\mathrm{Pol}(\mathcal{M})$ in terms of Wick powers and their composition in terms of a Wick-ordered product as discussed thoroughly in \cite{rejzner}. Yet, to this end, a second step is needed.

\begin{definition}\label{Def: Concrete Wick ordering}
     Let $\mathrm{Pol}(\mathcal{M})$ be the polynomial, local, microcausal functionals as per Definition \ref{Def: Local and Covariant algebra}. Given a Hadamard state $\Delta_+\in\mathcal{D}^\prime(\mathcal{M}\times\mathcal{M})$ we call operator ordering map of $\mathrm{Pol}(\mathcal{M})$ in $\mathrm{Pol}_{\Delta_+}(\mathcal{M})$ the application 
    $$\alpha_{\Delta_+}:\mathrm{Pol}(\mathcal{M})\to\mathrm{Pol}_{\Delta_+}(\mathcal{M})$$
    where $\alpha_{\Delta_+}$ is defined as per Equation \eqref{Eq: Wick ordering} with $-H$ replaced by $\hbar\Delta_+$. This is a $*$-isomorphism, namely, for all $F,G\in\mathrm{Pol}_{\Delta_+}(\mathcal{M})$,
    \begin{equation}\label{Eq: Wick products}
    F\star_{\Delta_+}G=\alpha_{\Delta_+}(\alpha^{-1}_{\Delta_+}(F) \, \alpha^{-1}_{\Delta_+}(G)).
    \end{equation}
\end{definition}

\noindent As a consequence of this last definition, we can work with abstract Wick polynomials $:A:_H$ as per Equation \eqref{Eq: Wick ordering} using the map $\alpha_{\Delta_+}$ whenever we wish to recover the standard notion that we are used to in quantum field theory.

\begin{example}\label{Exam: Wick ordering}
In order to clarify this last statement, we discuss an explicit example, namely, given $f\in\mathcal{D}(\mathcal{M})$, we consider the functional $\Phi^2_f$ as per Example \ref{Exam: Wick polynomials} and its abstract Wick ordered counterpart $:\Phi^2_f:_H=\alpha_{-H}(\Phi^2_f)$, see Definition \ref{Def: Wick ordering}. Applying in turn Definition \ref{Def: Concrete Wick ordering} we end up that, for all $\phi\in\mathcal{E}(\mathcal{M})$,
$$\alpha_{\Delta_+}(:\Phi^2_f:_H)(\phi)=\alpha_{\Delta_+-H}(\Phi^2_f)(\phi)=\Phi^2_f(\phi)+\int\limits_{\mathcal{M}}d\mu_x W(x,x)f(x),$$
where $W\in C^\infty(\mathcal{M}\times\mathcal{M})$ is the smooth function whose expression in any geodesically convex neighborhood abides by Equation \eqref{Eq: Hadamard two-point function}. Evaluating this expression at the configuration $\phi=0$ as per Remark \ref{Rem: Expectation value} yields the standard expression for the expectation value of the Wick ordered squared scalar field on a Hadamard state. Furthermore, using Equation \eqref{Eq: Wick products}, we can also recover the standard structure of the product of two Wick powers, namely
    \begin{equation*}
\Phi^2_f(\phi) \star_{\Delta_+} \Phi^2_{f'}(\phi) = \Phi^2_f(\phi) \Phi^2_{f'}(\phi)+4\hbar\Delta_+(f,f^\prime)\Phi_f(\phi) \Phi_{f'}(\phi)+2\hbar^2\Delta_+^2(f,f^\prime)
\end{equation*}
In the standard formulation of QFT, the product between Wick polynomials is presented as the implementation of contractions between fields. This pictorial interpretation is particularly useful when dealing with long computations and, hence, we remark for future convenience that the same convention can be adopted in the algebraic formulation:
\begin{equation*}
 \Phi^2_f(\phi) \star_{\Delta_+} \Phi^2_{f'}(\phi) = \Phi^2_f(\phi) \Phi^2_{f'}(\phi) + 4 \hbar \wick{
        \c1 \Phi^2_f(\phi)
        \c1 \Phi^2_{f'}(\phi) } +
 + 2 \hbar^2 \wick{
        \c1 \Phi_f(\phi) \c2 \Phi_{f}(\phi)
        \c1 \Phi_{f'}(\phi) \c2 \Phi_{f'}(\phi)},
\end{equation*}
where each contraction is performed by means of $\Delta_+$. 
\end{example}

\subsection{Interacting Theories in pAQFT}\label{Sec: Interacting Theories in pAQFT}
In the previous section we have discussed the building blocks necessary to analyze non-interacting field theories. If we wish to account also for non-linear interactions, additional structures are necessary. We shall discuss them in the following and, as starting point, we focus our attention once more on a real, scalar field $\phi$. We denote by $S$ the underlying total action and, in order to isolate the interaction potential $V$, we decompose it as 
    \begin{equation}\label{Eq: Action Splitting}
        S[\phi] = S_0[\phi] + \lambda V[\phi],
    \end{equation}
where $S_0$ is the action of the underlying free system. In the case under scrutiny and working already with the associated local functionals we shall consider
\begin{equation}\label{Eq: Action Terms}
S_0[\Phi]_f(\phi)=-\int\limits_{\mathcal{M}}d\mu_x\, f(x)(\frac{1}{2}g^{\mu\nu}\partial_\mu\phi\partial_\nu\phi+\frac{1}{2}(m^2+\xi R)\phi^2)\;\textrm{and}\; V[\Phi]_f(\phi)= - \frac{1}{n!}\int\limits_{\mathcal{M}}d\mu_x\, \phi^n(x) f(x),
\end{equation}
where $n=3,4$. 

\begin{remark}
    In Equation \eqref{Eq: Action Splitting} the splitting between a free and interacting contribution has a degree of arbitrariness, since nothing prevents us from incorporating the terms in the action, quadratic in the field, partly or completely in $V[\Phi]$. Such freedom does not alter globally the results of our analysis, but this statement relies on a rather technical and important property known as principle of perturbative agreement. We refer to \cite{drago} for a discussion of this topic.
\end{remark}

In order to encode the information of the non-linear potential in Equation \eqref{Eq: Action Splitting} with the perturbative approach to quantum field theory the building blocks are the $S$-matrix and the  Bogoliubov map. We shall give a succinct overview of the approach developed in \cite{brunetti2,deutsch,rejzner}.

The starting point is the {\em time ordered product} which is here defined in terms of a time ordering map $\mathcal{T}$ acting on the space of multi-local, polynomials functionals, which are tensor products of elements lying in $\mathrm{Pol}_{loc}(\mathcal{M})$ as in Definition \ref{Def: local functionals}. More precisely $\mathcal{T}$ is constructed out of a family of multi-linear maps
\[ \mathcal{T}_n : \mathrm{Pol}_{loc}^{\otimes n}(\mathcal{M}) \rightarrow
   \mathrm{Pol}_{\mu c}(\mathcal{M}), \]
which satisfy the constraints $\mathcal{T}_0 = 1$ and $\mathcal{T}_1 = \mathrm{id}$. The link between $\mathcal{T}$ and $\mathcal{T}_n$ is codified by the identity
\begin{equation}
	\mathcal{T} \left(
	\prod_{j = 1}^n F_j \right) =\mathcal{T}_n \left( \bigotimes_{j = 1}^n F_j
	\right).
\end{equation} 
In addition, one requires the maps $\mathcal{T}_n$ to be such that $\mathcal{T}$ is symmetric and to satisfy a causal factorization property. This can be stated as follows: consider $\{F_i \}_{i=1,\ldots, n}, \{G_j \}_{j=1,\ldots, m} \subset\mathrm{Pol}_{loc}(\mathcal{M})$ two arbitrary families of local, polynomial functionals such that $F_i \gtrsim G_j$ for any $i, j$. The symbol $\gtrsim$ entails that $\mathrm{supp} (F_i) \cap J^- (\mathrm{supp} (G_j)) = \emptyset$, where the support of a functional is as per Equation \eqref{Eq: support of a functional}. It descends that
\begin{equation}\label{Eq:factoriz-time-ordering}
  \mathcal{T} \left( \bigotimes_i F_i  \bigotimes_j G_j \right) =\mathcal{T}
  \left( \bigotimes_i F_i \right) \star_K \mathcal{T} \left( \bigotimes_j G_j
  \right), 
\end{equation}
where $\star_K$ is defined in Equation \eqref{starprod}. In the following we shall make a concrete choice for the kernel $K$. Yet the hypotheses at the heart of Equation \eqref{Eq:factoriz-time-ordering} might not be satisfied. In this case, one needs to devise a suitable extension criterion for $\mathcal{T}$, which requires in turn a renormalization procedure. In the algebraic approach to quantum field theory, this is encoded in the Epstein-Glaser inductive procedure \cite{epstein, hollands1}. Since we will not need to delve into the details of this aspect of pAQFT, we shall not discuss it further.

On the contrary, we focus on the following key fact: if we work with the algebra $\mathrm{Pol}_{\Delta_+}(\mathcal{M})$ as per Definition \ref{quanmicroal}, an explicit realization of the time-ordering map is completely characterized by the identity
\begin{equation}\label{Eq: time-ordered product}
	 \mathcal{T}^{\hbar H_F}  (F_1 \otimes \ldots \otimes F_n):=
	  F_1 \star_{F} \ldots \star_{F} F_n =
	\mathsf{M} \circ e^{\sum_{\ell < j}^n D^{\ell j}_{\hbar H_F}}  (F_1
	\otimes \ldots \otimes F_n), 
\end{equation}
where $\mathsf{M}$ denotes the pullback on $\textrm{Pol}_{\mu c}(\mathcal{M})\otimes
	\textrm{Pol}_{\mu c}(\mathcal{M})$ via the diagonal map
	\begin{align*}
		\iota :\, &\mathcal{E} (\mathcal{M})\rightarrow \mathcal{E} (\mathcal{M}) \times \mathcal{E} (\mathcal{M})\\
		&\phi \longrightarrow\iota (\phi)=(\phi, \phi).
	\end{align*}
while $F_i\in\mathrm{Pol}_{\mu c}(\mathcal{M})$, for all $i=1,\dots,n$. Here $H_F\in\mathcal{D}'(\mathcal{M}\times\mathcal{M})$ denotes the Feynman parametrix as in Equation \eqref{Eq: Feynman parametrix}, while 
$$ D^{\ell j}_{\hbar H_F}:= \langle \hbar H_F, \frac{\delta}{\delta \phi_{\ell}}\otimes\frac{\delta}
   {\delta \phi_j}\rangle,$$
which is manifestly symmetric under exchange of $j$ and $\ell$.

\subsubsection{The S-matrix and the Bogoliubov map.}\label{Sec: Bogoliubov map}
\noindent Starting from Equation \eqref{Eq: time-ordered product}, we can define the key structures at the heart of our analysis. Although the whole procedure can be applied to a large class of interacting field theories, for definiteness, we shall always focus on a theory ruled by Equation \eqref{Eq: Action Splitting}, considering polynomial functionals.

\begin{definition}[S-matrix]
\label{smat}
    Let $\mathcal{M}$ be a globally hyperbolic spacetime and let $\lambda V \in \mathrm{Pol}_{loc}(\mathcal{M})$ be the interacting potential. If we denote by $\star_F$ the time-ordered product between local functionals, we call \emph{S-matrix} the functional
    \begin{flalign}
        \mathcal{S}(\lambda V) &:= \sum_{n \ge 0} \frac{1}{n!} \left( \frac{i \lambda}{\hbar} \right)^n \underbrace{V \star_F ... \star_F V}_{n},\label{Eq: S-Matrix}
    \end{flalign}
    where $\lambda \in \mathbb{R}$ is the coupling constant. Furthermore, we denote by $\mathrm{Pol}_{\mu c}[[\lambda]](\mathcal{M})$ the collection of all functionals defined as a formal power series in the parameter $\lambda$, such that $F\in \mathrm{Pol}_{\mu c}[[\lambda]](\mathcal{M})$ if $F=\sum\limits_{n=0}^\infty\lambda^n F_n$ with $F_n\in\mathrm{Pol}_{\mu c}(\mathcal{M})$.
\end{definition}

\noindent Observe that the S-matrix is also a Laurent series with respect to $\hbar$. Since, in the following, we are more interested in controlling a perturbative expansion with respect to the coupling constant $\lambda$, we suppress this dependency, leaving it implicit.

The S-matrix admits an inverse in the sense of formal power series. It is useful for later purposes to express it explicitly by means of the \emph{anti-Feynman parametrix}, \textit{i.e.}, 
\begin{equation}\label{Eq: Anti-Feynman Parametrix}
    H_{AF} := H - i \Delta_R = H_F^*\in\mathcal{D}^\prime(\mathcal{M}\times\mathcal{M}),
\end{equation}
where $\Delta_R \in \mathcal{D}^\prime(\mathcal{M} \times \mathcal{M})$ denotes the retarded fundamental kernel associated to $P_0$ as in Equation \eqref{Eq: KG}, while $H \in \mathcal{D}'(\mathcal{M} \times \mathcal{M})$ is any bi-distribution differing from the two-point correlation function of a Hadamard state by a smooth remainder. The identity $H_{AF}=H^*_F$ should be read as a statement that $H_{AF}$ is the formal adjoint of $H_F$. Thus, we can write the inverse of $\mathcal{S}(\lambda V)$ as
\begin{equation}
\label{invsmat}
    \mathcal{S}^{\star -1} (\lambda V) := \sum_{n \ge 0} \frac{1}{n!} \left( - \frac{i \lambda}{\hbar} \right)^n \underbrace{V \star_{AF} ... \star_{AF} V}_{n}, 
\end{equation}
where $\star_{AF}$ denotes the product defined via Equation \eqref{starprod} using $H_{AF}$ in place of $K$. In addition the superscript $\star$ indicates that $\mathcal{S}^{\star -1}$ is the inverse of the S-matrix in the sense of a formal power series with respect to the product induced by $\star_H$, namely, denoting by $\textbf{1}$ the identity functional
$$\mathcal{S}(\lambda V)\star_H\mathcal{S}^{\star -1} (\lambda V)=\mathcal{S}^{\star -1} (\lambda V)\star_H\mathcal{S}(\lambda V)=\textbf{1}.$$

\noindent The second key ingredient is the \emph{Bogoliubov map}, which allows to represent the interacting observables in terms of their free counterparts. We define it specifically in the scenario under investigation.

\begin{definition}[Bogoliubov map] Let $\mathcal{M}$ be a four-dimensional globally hyperbolic spacetime and let $V \in \mathrm{Pol}_{loc}(\mathcal{M})$ be a potential describing an interaction of a real scalar field thereon, whose dynamics is ruled by the action in Equation \eqref{Eq: Action Splitting}. Let $\mathcal{S}$ be the S-matrix as per Definition \ref{smat} and let its inverse be as per Equation \eqref{invsmat}. Given $F\in \mathrm{Pol}_{loc}(\mathcal{M})$, we call \emph{Bogoliubov map},
\begin{equation}\label{Eq: Bogoliubov map}
    R_{\lambda V}(F):= \mathcal{S}^{\star -1}(\lambda V) \star_H [\mathcal{S}(\lambda V) \star_F F], \, \, \forall F \in \mathrm{Pol}_{loc}(\mathcal{M}). 
\end{equation}
where $\star_{H}$ and $\star_F$ are defined as in Equation \eqref{starprod} with the kernel $K$ replaced by $H$ and by $H_F$ respectively.
\end{definition}

\noindent Observe that, given a generic free observable $F \in \mathcal{F}_{loc}(\mathcal{M})$, the interacting counterpart $R_{\lambda V}(F)\in \mathcal{F}_{\mu \, c}[[\lambda]](\mathcal{M})$ is as a formal power series in the coupling constant $\lambda$ as well as a Laurent series in $\hbar$.

\subsubsection{Examples}\label{Sec: Examples}
In this section, we discuss an instructive example which shows how to explicitly compute the interacting fields via the Bogoliubov map. A reader who is already acquainted with the framework can safely skip this subsection. Under the same hypotheses of the preceding analysis, we consider the observables
\begin{equation}\label{Eq: Examples of Functionals}
    \Phi_f[\phi] := \int_{\mathcal{M}}\,d\mu_x \phi(x) f(x), \, \, \, \,  \Phi^2_f[\phi] := \int_{\mathcal{M}}d\mu_x\, \phi^2(x) f(x), \, f \in \mathcal{D}(\mathcal{M})\;\textrm{and}\;\phi\in\mathcal{E}(\mathcal{M})
\end{equation}
as well as an interaction potential 
\begin{equation*}
    V_h [\phi] := - \frac{1}{n!} \int_{\mathcal{M}}d\mu_x\, \phi^n(x) h(x), \, n = 3,4, \, h \in \mathcal{D}(\mathcal{M}). 
\end{equation*}
To make contact with the analysis in Section \ref{Sec: Interacting Stress-Energy Tensor}, we are interested in computing the expectation value of the interacting counterparts of Equation \eqref{Eq: Examples of Functionals} at second order in $\lambda$, namely $R_{\lambda V}(\Phi_f) \vert_{\phi = 0}$ and $R_{\lambda V}(\Phi^2_f) \vert_{\phi = 0}$.  As pointed out in Remark \ref{Rem: Expectation value}, in this formalism, if we work with states rather then with parametrices the expectation value of any observable corresponds to the evaluation at the configuration $\phi = 0$. For this reason and, in preparation to Section \ref{Sec: Interacting Stress-Energy Tensor}, we shall consider the evaluation at $\phi=0$ even though all deformations are considered with reference to parametrices. We start from the S-matrix in Equation \eqref{Eq: S-Matrix}, expanding it as
\begin{flalign}
    \mathcal{S}(\lambda V) & = \mathbf{1} + \frac{i \lambda}{\hbar} V - \frac{1}{2} \frac{\lambda^2}{\hbar^2} V \star_F V + \mathcal{O}(\lambda^3), \\
     \mathcal{S}^{\star -1}(\lambda V) & = \mathbf{1} - \frac{i \lambda}{\hbar} V - \frac{1}{2} \frac{\lambda^2}{\hbar^2} V \star_{AF} V + \mathcal{O}(\lambda^3).
\end{flalign}
Focusing on the interacting field expanded up to order $\mathcal{O}(\lambda^3)$, we obtain 
\begin{flalign}
    R_{\lambda V}(\Phi_f) &=_{\mathcal{O}(\lambda^3)} \Phi_f + \frac{i \lambda}{\hbar} \left( V_h \star_F \Phi_f - V_h \star_H \Phi_f \right) \notag \\ & -\frac{\lambda^2}{2 \hbar^2} \left[ (V_h \star_{AF} V_h) \star_H \Phi_f + (V_h \star_F V_h) \star_F \Phi_f - 2 V_h \star_H (V_h \star_F \Phi_f)\right],\label{Eq: First Order Expansion}
\end{flalign}
and
\begin{flalign}
    R_{\lambda V}(\Phi^2_f) &=_{\mathcal{O}(\lambda^3)} \Phi^2_f + \frac{i \lambda}{\hbar} \left( V_h \star_F \Phi^2_f - V_h \star_H \Phi^2_f \right) \notag \\ & -\frac{\lambda^2}{2 \hbar^2} \left[ (V_h \star_{AF} V_h) \star_H \Phi^2_f + (V_h \star_F V_h) \star_F \Phi^2_f - 2 V_h \star_H (V_h \star_F \Phi^2_f)\right], \label{Eq: Second Order Expansion}
\end{flalign}
where $=_{\mathcal{O}(\lambda^3)}$ entails a truncation at second order in the coupling constant $\lambda$. Hence, only terms which are not proportional to $\Phi^k$, $k\geq 1$, do contribute and we refer to these elements as \emph{maximally-contracted addenda}. Although in the examples discussed, one can infer by direct inspection which are these terms, for more complicated functionals, it is convenient to adopt a graphical formalism as in Example \ref{Exam: Wick ordering}, without making all computations explicit. To outline its implementation, let us consider $R_{\lambda V} (\Phi_f)$ as per Equation \eqref{Eq: First Order Expansion} and observe that  
\begin{itemize}
    \item at first order in $\lambda$, all contributions vanish when $\phi = 0$, since they are linear in the field; 
    \item no maximally contracted terms do exist, even at second order. 
\end{itemize}
As a consequence $R_{\lambda V} (\Phi_f) \vert_{\phi = 0} = \mathcal{O}(\lambda^3)$. Actually, an inductive argument yields that $R_{\lambda V}(\Phi_f)\vert_{\phi = 0} = 0$ at any order in perturbation theory, since all the admissible contractions among the fields always leave an unpaired term. Focusing instead on $R_{\lambda V}(\Phi^2_f)$ as per Equation \eqref{Eq: Second Order Expansion}, it is possible to infer that 
\begin{itemize}
    \item none of the terms of the first order expansion survives. In fact, in $V_h \star \Phi_f^2$ as well as in $V_h \star_F \Phi_f^2$, it is possible to perform at most two contractions, hence leaving at least one power of $\Phi_f$ remaining,
    \item at second order in perturbation theory, from $(V_h \star_{AF} V_h) \star \Phi_f^2$, it is possible to obtain 
     \begin{equation*}
        \wick{
        (\c1 \Phi_h \c2 \Phi_{h} \c3 \Phi_{h} \c4 \Phi_{h} \, \, \;
        \c2 \Phi_{h} \c3 \Phi_{h} \c4 \Phi_{h} \c5 \Phi_{h}) \, \,
        \c1 \Phi_{f} \c5 \Phi_{f}}
    \end{equation*}
    in the case $n=4$ and 
    \begin{equation*}
         \wick{
        (\c1 \Phi_h \c2 \Phi_{h} \c3 \Phi_{h} \, \, \;
        \c2 \Phi_{h} \c3 \Phi_{h} \c5 \Phi_{h}) \, \,
        \c1 \Phi_{f} \c5 \Phi_{f}}
    \end{equation*}
    in the case $n=3$. Since the contractions between elements lying within the round brackets are implemented by $H_{AF}$, see Equation \eqref{Eq: Anti-Feynman Parametrix}, while all others by the Hadamard parametrix $H$, see Equation \eqref{Eq: Hadamard parametrix}, at the level of integral kernels we obtain
    \begin{equation}
        - \frac{\lambda^2 \hbar^{(n-1)}}{(n-1)!} \int_{\mathcal{M}^3}d \mu_x d \mu_y d \mu_z\, h(x) h(y) f(z) H^{(n-1)}_{AF} (x,y) H(x,z) H(y,z).
    \end{equation}
    \item Focusing instead on $(V_h \star_F V_h) \star_F \Phi^2_f$ in Equation \eqref{Eq: Second Order Expansion}, at a graphical level the same contractions contribute. However, they are related to different kernels and therefore the analytic expression becomes
    \begin{equation}
         - \frac{\lambda^2 \hbar^{(n-1)}}{(n-1)!} \int_{\mathcal{M}^3}d \mu_x d \mu_y d \mu_z\, h(x) h(y) f(z) H^{(n-1)}_{F} (x,y) H_F(x,z) H_F(y,z).
    \end{equation}
    \item The last remaining contribution in Equation \eqref{Eq: Second Order Expansion} is $V_h \star_H (V_h \star_F \Phi^2_f)$, which yields at a graphical level
    \begin{equation*}
        \wick{
        \c1 \Phi_h \c2 \Phi_{h} \c3 \Phi_{h} \c4 \Phi_{h} \, \, 
        (\c2 \Phi_{h} \c3 \Phi_{h} \c4 \Phi_{h} \c5 \Phi_{h}) \, \,
        \c1 \Phi_{f} \c5 \Phi_{f}}
    \end{equation*}
    in the case $n =4$ while, if $n=3$, 
    \begin{equation*}
        \wick{
        \c1 \Phi_h \c2 \Phi_{h} \c3 \Phi_{h} \, \, 
        (\c1 \Phi_{h} \c2 \Phi_{h} \c4 \Phi_{h}) \, \,
        \c3 \Phi_{f} \c4 \Phi_{f}}
    \end{equation*}
    At the level of integral kernels this translates as 
    \begin{equation*}
       - \frac{\lambda^2 \hbar^{(n-1)}}{(n-1)} \int_{\mathcal{M}^3} d\mu_x d\mu_y d\mu_z \, h(x) h(y) f(z) H^{(n-1)}(x,y) H(x,z) H_F(y,z). 
    \end{equation*}
\end{itemize}
Collecting all contributions, we end up with
\begin{flalign}
    R_{\lambda V} (\Phi^2_f) \vert_{\phi = 0} = & \frac{\lambda^2 \hbar^{(n-1)}}{(n-1)!} \int_{\mathcal{M}^3}d\mu_x d\mu_y d\mu_z\,  h(x) h(y) f(z) [2H^{(n-1)}(x,y) H(x,z) H_F(y,z) \notag \\  &- H^{(n-1)}_F(x,y) H_F(x,z) H_F(y,z) - H^{(n-1)}_{AF}(x,y) H(x,z) H(y,z) ]. \label{Eq: Perturbative Squared Field}
\end{flalign}

\begin{remark}\label{Rem: Renormalization ambiguities in the example}
    Observe that Equation \eqref{Eq: Perturbative Squared Field} is merely formal at this stage, since, if $n=3,4$, we have to deal with powers of the parametrices $H,H_F$ and $H_{AF}$. While Equation \eqref{Eq: WF Delta+}, combined with the identity $\mathrm{WF}(\Delta_+)|_{T^*(\mathcal{O}\times\mathcal{O})}=\mathrm{WF}(H)$ and with Theorem \ref{Thm: Hormander criterion} entails that $H^{(n-1)}\in\mathcal{D}^\prime(\mathcal{O}\times\mathcal{O})$, this is not the case for $H^{(n-1)}_F$ and for $H^{(n-1)}_{AF}$. Yet, as discussed in Remark \ref{Rem: Powers of HF}, $H^{(n-1)}_F\in\mathcal{D}^\prime(\mathcal{M}\times\mathcal{M}\setminus X)$, where $X$ is the collection of points $(x,y)\in\mathcal{M}\times\mathcal{M}$ connected by a lightlike
geodesic. Therefore an extension procedure is necessary and this leads to the well-known renormalization ambiguities. Since they do not play a key r\^{o}le in the derivation of the results of this paper, we do not discuss them in detail. Observe that, being $H_{AF}$ the formal adjoint of $H_F$, the same conclusion can be drawn for its powers.  
\end{remark}

\section{On the stress-Energy Tensor for a Free Theory}\label{Sec: Free Stress-Energy Tensor}
The time is ripe to focus our attention on the stress-energy tensor which encodes information about the energy and the momentum of the underlying fields. In this section we consider only a free field theory, whose dynamics is ruled by the operator \eqref{Eq: KG} and we discuss with the language of functionals the structure of the stress-energy tensor. In particular we investigate the conformal anomaly and we shall follow mainly the approach of \cite{moretti}, which will be extended in Section \ref{Sec: Interacting Stress-Energy Tensor} to an interacting scenario. At a classical level, the stress-energy tensor can be defined as
\begin{equation}\label{Eq: Definition of Stress-Energy Tensor}
        T_{\mu \, \nu} = \frac{-2}{\sqrt{|g|}} \frac{\delta S}{\delta g^{\mu \, \nu}}.
\end{equation}
If we consider the action $S_0$ in Equation \eqref{Eq: Action Terms}, one obtains 
\begin{flalign}
\label{seten}
    T_{\mu \, \nu}(x) &= \partial_{\mu} \phi(x) \partial_{\nu} \phi(x) - \frac{1}{2} g_{\mu \, \nu}(x) [\partial^{\rho} \phi(x) \partial_{\rho} \phi(x) + m^2 \phi^2(x)] \\ \notag &+ \xi G_{\mu \, \nu}(x) \phi^2(x) + \xi [g_{\mu \, \nu}(x) \Box - \nabla_{\mu} \nabla_{\nu}] \phi^2(x), 
\end{flalign}
where $G_{\mu \, \nu} = R_{\mu \, \nu} - \frac{1}{2} R g_{\mu \, \nu}$ is the Einstein tensor. We observe that, if we set $\xi=\frac{1}{6}$ and if $g$ is chosen as the metric of Minkowski spacetime, then Equation \eqref{seten} reduces to the improved stress-energy tensor introduced in \cite{coleman}.

\begin{remark} 
    The stress-energy tensor as per Equation \eqref{seten} can be interpreted as the integral kernel of a suitable functional following Remark \ref{Rem: Derivatives applied to a functional}. More precisely, let us introduce the following notation 
    \begin{equation*}
        (K_1 \Phi^k K_2 \Phi^{k^\prime})_f [\phi] := \int_{\mathcal{M}^2} K_1 \phi^k(x) K_2^\prime \phi^{k^\prime}(y) f(x) \delta_2(x,y) \, d \mu_x d \mu_y, \, \, \forall \phi \in \mathcal{E}(\mathcal{M}), 
    \end{equation*}
    where, $k,k^\prime\in\mathbb{N}$, $K_1$ and $K_2$ are arbitrary differential operators, while the superscript $\prime$ denotes that differentiation acts with respect to the $y-$variable. If we work at the level of integral kernels, we can write
    \begin{flalign}
        \textbf{T}_{\mu \, \nu}[\phi](z)  &=  (\partial_{\mu} \Phi \partial_{\nu} \Phi) [\phi](z) - \frac{1}{2} g_{\mu \, \nu}(z) (\partial^{\rho} \Phi \partial_{\rho} \Phi)[\phi](z) - \frac{m^2}{2} g_{\mu \, \nu}(z) \Phi^2 [\phi](z) \notag\\ 
        \label{Eq: Functional Stress-Energy Tensor} &+ \xi G_{\mu \, \nu}(z) \Phi^2[\phi](z) + \xi g_{\mu \, \nu}(z) (\Box \Phi^2)[\phi](z) - \xi (\nabla_{\mu} \nabla_{\nu} \Phi^2)[\phi](z). 
    \end{flalign}
It turns out that, once tested, $\textbf{T}_{\mu \, \nu}$ lies in $\mathrm{Pol}_{loc}(\mathcal{M})$:
    \begin{equation*}
        \textbf{T}_{\mu \, \nu, f}[\phi] = \int_{\mathcal{M}} T_{\mu \, \nu}(x) f^{\mu \, \nu}(x) \, d\mu_x, \, \, \forall \phi \in \mathcal{E}(\mathcal{M}),
    \end{equation*}
    where $f^{\mu \, \nu} \in \Gamma_0(T\mathcal{M}^{\otimes_s 2})$. Henceforth we shall only consider integral kernels, leaving the test function implicit and with a slight abuse of notation we write $\textbf{T}_{\mu \, \nu}\in \mathrm{Pol}_{loc}(\mathcal{M})$. 
\end{remark}

\noindent Equation \eqref{seten} and its functional counterpart, Equation \eqref{Eq: Functional Stress-Energy Tensor}, do not carry any information about an underlying quantization procedure and therefore they encode all classical properties of the stress-energy tensor, conservation laws in particular. The following proposition translates this fact:

\begin{proposition}[On-shell Conservation] 
Let $(\mathcal{M}, g)$ be a globally hyperbolic spacetime and let $\textbf{T}_{\mu \, \nu} \in \mathrm{Pol}_{loc}(\mathcal{M})$ be the stress-energy tensor as per Equation \eqref{Eq: Functional Stress-Energy Tensor} associated to a real scalar field whose dynamics is ruled by the Klein-Gordon operator $P_0$ as per Equation \eqref{Eq: KG}. Let
$$\mathcal{E}_S(\mathcal{M}):=\{\phi\in\mathcal{E}(\mathcal{M})\;|\;P_0\phi=0\}.$$
It holds that
\begin{equation*}
    \nabla^{\mu} \textbf{T}_{\mu \, \nu}[\phi_0](z) = 0, \, \, \forall \phi_0 \in \mathcal{E}_S(\mathcal{M}).
\end{equation*}
\end{proposition}

\noindent The proof of this statement is a direct application of the elementary operations on distributions combined with the dynamics of the underlying real scalar field. Hence we omit the details. Another key quantity that we wish to investigate is the trace of the stress-energy tensor which is the source of the {\em conformal anomaly}. If the spacetime dimension is $\dim\mathcal{M}=4$, then the classical trace reads
    \begin{equation}\label{Eq: Classical Trace}
        T[\phi_0](z) := g^{\mu \, \nu}(z) \textbf{T}_{\mu \, \nu}[\phi_0](z) = - m^2 \Phi^2[\phi_0](z) + \frac{1}{2} (6 \xi-1) (\Box \Phi^2) [\phi_0](z), \, \, \forall \phi_0 \in \mathcal{E}_S(\mathcal{M}), 
    \end{equation}
which vanishes if $m=0$ and $\xi=\frac{1}{6}$, namely for a massless, conformally coupled real scalar field.    

\begin{remark}\label{Rem: Non-Uniqueness}
In view of the discussion of the preceding section, the quantum counterpart of the stress-energy tensor is obtained representing the functional in Equation \eqref{Eq: Functional Stress-Energy Tensor} on $\mathrm{Pol}_{\Delta_+}(\mathcal{M})$, the quantum algebra of microcausal functionals. This is a two-step procedure. First of all one needs to construct an abstract Wick-ordered stress-energy tensor along the lines of Definition \ref{Def: Wick ordering}. Subsequently, the outcome must be represented on $\mathrm{Pol}_{\Delta_+}(\mathcal{M})$ following Definition \ref{Def: Concrete Wick ordering}. The first step is made necessary by the non-linear nature of the functional in Equation \eqref{Eq: Functional Stress-Energy Tensor}.

Yet, in the transition to the quantum realm, one needs to make sure that the structural features of the stress-energy tensor are preserved, in particular the conservation law codified by its property of being divergence free. As it is well-known in the literature, the construction of Wick polynomials using the Hadamard parametrix as in Equation \eqref{Eq: Wick ordering} is hard coding locality and covariance, but it can disrupt any property related to dynamics since $H$ is a bi-solution of the equation of motion only up to a smooth remainder. In particular this entails that the quantized stress-energy tensor associated to a free field theory might fail to be divergence free. A way out from this unacceptable option has been studied exploiting the existing freedom in the construction of Wick-ordered observables, see \cite{hollands}. 

An alternative, albeit ultimately equivalent, procedure has been outlined in \cite{moretti} and we shall abide by it in this work. More precisely, already, at a classical level, there exists a freedom in defining the stress-energy tensor consisting of adding a term proportional to the equation of motion,  automatically conserved and traceless on shell. In other words we introduce 
\begin{flalign}
\label{improvset}
    T_{\mu \, \nu}^{(\eta)} (z) &= \partial_{\mu} \phi(z) \partial_{\nu} \phi(z) - \frac{1}{2} g_{\mu \, \nu}(z) [\partial^{\rho} \phi(z) \partial_{\rho} \phi(z) + m^2 \phi^2(z)] \\ \notag &+ \xi G_{\mu \, \nu}(z) \phi^2(z) + \xi [g_{\mu \, \nu}(z) \Box - \nabla_{\mu} \nabla_{\nu}] \phi^2(z) + \eta g_{\mu \, \nu}(z) \phi(z) (P_0 \phi)(z), 
\end{flalign}
where $\eta \in \mathbb{R}$ is a free parameter while $P_0$ is the Klein-Gordon operator as in Equation \eqref{Eq: KG}. 
\end{remark}

\noindent Switching our attention back to the functional formalism, we can interpret Equation \eqref{improvset} as the integral kernel associated to the functional $\textbf{T}_{\mu \, \nu}^{(\eta)}\in\mathrm{Pol}_{loc}(\mathcal{M})$ such that
\begin{flalign}       
\textbf{T}_{\mu \, \nu}^{(\eta)}[\phi](z)  &=  (\partial_{\mu} \Phi \partial_{\nu} \Phi) [\phi](z) - \frac{1}{2} g_{\mu \, \nu}(z) (\partial^{\rho} \Phi \partial_{\rho} \Phi) [\phi](z) - \frac{m^2}{2} g_{\mu \, \nu}(z) \Phi^2 [\phi](z) \notag\\ 
        \label{Eq: Functional Extra Stress-Energy Tensor} &+ \xi G_{\mu \, \nu}(z) \Phi^2[\phi](z) + \xi g_{\mu \, \nu}(z) (\Box \Phi^2)[\phi](z) - \xi (\nabla_{\mu} \nabla_{\nu} \Phi^2) [\phi](z)+\eta g_{\mu\nu}(z)(\Phi\,P_0\Phi)(\phi)(z), 
\end{flalign}
where $\phi\in\mathcal{E}(\mathcal{M})$. The following theorem is a slavish translation in the language of functionals of the result proven in \cite{moretti}. Here, for definiteness, we focus our attention to four-dimensional backgrounds, in order to make contact with the scenarios of interest in theoretical physics.

\begin{theorem}
\label{valter}
    Let $(\mathcal{M}, g)$ be a $4-$dimensional globally hyperbolic spacetime and let $\Delta_+$ be a Hadamard two-point correlation function for a real scalar fields whose dynamics is ruled by the Klein-Gordon operator $P_0$ as in Equation \eqref{Eq: KG}. Let $\wickk{\textbf{T}_{\mu \, \nu}^{(\eta)}}(z)=\alpha_{-H}(\textbf{T}_{\mu \, \nu}^{(\eta)})(z)$ be the integral kernel of the Wick-ordered stress-energy tensor as per Equation \eqref{Eq: Wick ordering}. It holds that
    \begin{itemize}
        \item[1.] $\nabla^\mu\wickk{\textbf{T}_{\mu \, \nu}^{(\eta)}}(z)=0$ if and only if $\eta = \frac{1}{3}$; 
        \item[2.] the functional $\wickk{g_{\mu\nu}(\Phi\,P_0\Phi)}(z)=g_{\mu\nu}(z)\alpha_{-H}((\Phi\,P_0\Phi))(z)$ is proportional to the identity functional, namely, for every $f\in\mathcal{D}(\mathcal{M})$
        \begin{equation*}
            \wickk{(\Phi P_0 \Phi)}(z) = \frac{3}{4 \pi^2} v_1(z,z) \textbf{1}, 
        \end{equation*}
        where $v_1$ is the first coefficient in the Hadamard expansion of Equation \eqref{Eq: Hadamard expansion}.
        \item[3.] if we set $\eta = \frac{1}{3}$ and $\xi = \frac{1}{6}$, then the trace of the stress-energy tensor reads 
        \begin{equation}
        \label{traceT}
            \wickk{\textbf{T}}(z)= g^{\mu\nu}(z)\wickk{\textbf{T}_{\mu \, \nu}^{(\frac{1}{3})}}(z) = -m^2 \wickk{\Phi^2}(z) + \frac{1}{4 \pi^2} v_1(z,z) \textbf{1}.
        \end{equation}
        The coefficient proportional to the integral of $v_1(z,z)$ is called {\bf conformal (or trace) anomaly}.
    \end{itemize}
\end{theorem} 


The trace anomaly is a physical result, in the sense that it is state independent and therefore it is a general feature of the quantized stress energy tensor. Using the Hadamard recursion relations, it is possible to evaluate $v_1$ in terms of geometric quantities and, setting $\xi=\frac{1}{6}$, -- see for example \cite{decanini}
\begin{equation}\label{Eq: v1}
    v_1(z,z) = \frac{m^4}{8}+\frac{1}{720} \left[C_{\mu\nu\rho\sigma}(z)C^{\mu\nu\rho\sigma}(z) + R_{\mu\nu}(z)R^{\mu\nu}(z)-\frac{R^2}{3}(z)+\Box R(z)\right], 
\end{equation}
where $C_{\mu\nu\rho\sigma}$ and $R_{\mu\nu\rho\sigma}$ are respectively the Weyl and the Riemann tensor built out the metric $g$.

\section{Stress-Energy tensor for a self-interacting theory}\label{Sec: Interacting Stress-Energy Tensor}


In this section we consider a real scalar field on a globally hyperbolic spacetime $(\mathcal{M},g)$ whose dynamics is ruled by the classical action in Equation \eqref{Eq: Action Splitting} where $n$ is either $3$ or $4$, while $\dim\mathcal{M}=4$. The associated stress-energy tensor can be computed using Equation \eqref{Eq: Definition of Stress-Energy Tensor} as
\begin{flalign}
    \mathcal{T}_{\mu \, \nu}(z) &= \partial_{\mu} \phi(z) \partial_{\nu} \phi(z) - \frac{1}{2} g_{\mu \, \nu} (z) [\partial^{\rho} \phi(z) \partial_{\rho} \phi(z) - m^2 \phi^2(z)] \notag\\  & - \frac{\lambda}{n!} g_{\mu \, \nu}(z) \phi^n(z) + \xi G_{\mu \, \nu}(z) \phi^2(z) + \xi [g_{\mu \, \nu}(z) \Box - \nabla_{\mu} \nabla_{\nu}]\phi^2(z). \label{Eq: Interacting Stress-Energy Tensor}
\end{flalign}

We shall investigate the quantum counterpart of Equation \eqref{Eq: Interacting Stress-Energy Tensor} following the same procedure used in Section \ref{Sec: Free Stress-Energy Tensor} for the free field theory counterpart. More precisely, we read $\mathcal{T}_{\mu \, \nu}(z)$ as the integral kernel of a functional $\boldsymbol{\mathcal{T}}_{\mu \, \nu}\in\mathrm{Pol}_{loc}(\mathcal{M})$ which reads for every $\phi\in\mathcal{E}(\mathcal{M})$ at the level of integral kernel
\begin{flalign}
    \boldsymbol{\mathcal{T}}_{\mu \, \nu}[\phi](z) &= (\partial_{\mu} \Phi \partial_{\nu} \Phi)[\phi](z) - \frac{1}{2} g_{\mu \, \nu}(z) [(\partial^{\rho} \Phi \partial_{\rho} \Phi)[\phi](z) - m^2 \Phi^2[\phi](z)] \notag\\  & - \frac{\lambda}{n!} g_{\mu \, \nu}(z)\Phi^n[\phi](z) + \xi G_{\mu \, \nu}(z) \Phi^2[\phi](z) + \xi [[g_{\mu \, \nu}(z) \Box - \nabla_{\mu} \nabla_{\nu}]\Phi^2](\phi)(z). \label{Eq: Functional Interacting Stress-Energy Tensor}
\end{flalign}

\noindent As in the free case analyzed in the previous section, Equation \eqref{Eq: Interacting Stress-Energy Tensor} and its functional counterpart, Equation \eqref{Eq: Functional Interacting Stress-Energy Tensor}, do not carry any information of an underlying quantization procedure and therefore they encode all classical properties of the stress-energy tensor, conservation laws in particular. The following proposition translates this fact:

\begin{proposition}[On-shell Conservation - Interacting Case] 
Let $(\mathcal{M}, g)$ be a globally hyperbolic spacetime and let $\boldsymbol{\mathcal{T}}_{\mu \, \nu} \in \textrm{Pol}_{loc}(\mathcal{M})$ be the stress-energy tensor as per Equation \eqref{Eq: Functional Stress-Energy Tensor} associated to a real scalar field whose dynamics is ruled by Equation \eqref{Eq: non-linear dynamics}. Let
$$\mathcal{E}_I(\mathcal{M}):=\{\phi\in\mathcal{E}(\mathcal{M})\;|\;P_0\phi=\frac{\lambda}{(n-1)!}\phi^{n-1}\}.$$
It holds that
\begin{equation*}
    \nabla^\mu\boldsymbol{\mathcal{T}}_{\mu \, \nu}[\phi_0](z) = 0, \, \, \forall \phi_0 \in \mathcal{E}_I(\mathcal{M}).
\end{equation*}
\end{proposition}

If we represent the functional in Equation \eqref{Eq: Interacting Stress-Energy Tensor} on $\mathrm{Pol}_{\Delta_+}(\mathcal{M})$, the quantum algebra of microcausal functionals, we face the same hurdles as in the free case. More precisely the Wick-ordered interacting stress-energy tensor is no longer divergence free. In making this statement precise, there is an additional complication which needs to be accounted for, namely the underlying dynamics is non-linear. Therefore, at a quantum level, one needs to replace the classical field with the quantum, interacting counterpart. As shown in Section \ref{Sec: Bogoliubov map} this is constructed via the Bogoliubov map, which yields a formal power series in the coupling constant $\lambda$. Hence, all statements, even the conservation of the stress-energy tensor, have to be understood in perturbation theory. 

In the following we shall prove that the framework introduced in \cite{moretti} can be adapted also to the interacting scenario. Therefore we shall individuate a modified stress-energy tensor showing that its quantum counterpart is divergence-free up to second order in perturbation theory. As a by-product also the quantum trace has to be modified accordingly. \\

Following the same line of reasoning outlined in Section \ref{Sec: Free Stress-Energy Tensor}, we exploit that the classical stress-energy tensor is not uniquely defined and one has the freedom of adding a contribution which, on shell, is both divergence free and traceless:
\begin{equation*}
    \eta g_{\mu \, \nu}(x) \left( \phi(x) P_0 \phi(x) - \frac{\lambda}{(n-1)!} \phi^{n}(x)\right).
\end{equation*}
At the level of functionals, this entails that we must consider $\boldsymbol{\mathcal{T}}^{\eta}_{\mu \, \nu} \in \textrm{Pol}_{loc}(\mathcal{M})$ such that
\begin{equation}\label{Eq: Extended Intracting Quantum Stress-Energy Tensor}
\boldsymbol{\mathcal{T}}^{\eta}_{\mu \, \nu}(z)=\boldsymbol{\mathcal{T}}_{\mu \, \nu}(z)+\eta g_{\mu\nu}(z)\left((\Phi P_0\Phi)(z)-\frac{\lambda}{(n-1)!}\Phi^{n}(z)\right).
\end{equation}

In order to establish which is the value of $\eta$ which yields a conserved quantum stress energy tensor, we need to represent Equation \eqref{Eq: Extended Intracting Quantum Stress-Energy Tensor} on $\mathrm{Pol}_{\Delta_+}(\mathcal{M})$. This requires in turn the application of the Bogoliubov map in Equation \ref{Eq: Bogoliubov map} to construct the interacting quantum observable. To this end, a direct inspection of Equation \eqref{Eq: Extended Intracting Quantum Stress-Energy Tensor} compared to the explicit form of the Bogoliubov map unveils that, in the construction of the quantum, interacting stress-energy tensor, the terms appearing all have the same algebraic form. In the following lemma we make explicit how we can evaluate them concretely. 

\begin{lemma}\label{Lem: Il conto si fa così}
Let $(\mathcal{M}, g)$ be a $4-$dimensional globally hyperbolic spacetime. Let 
\begin{equation*}
\begin{cases}
    D= a(x) \nabla_{\mu_1} ... \nabla_{\mu_i}, \, \, a \in C^{\infty}(\mathcal{M}), \\
    D^\prime = b(x) \nabla_{\nu_1} ... \nabla_{\nu_j}, \, \, b \in C^{\infty}(\mathcal{M})
\end{cases}
\end{equation*}
be finite-order differential operators constructed from the covariant derivative built out of $g$. Let $F\in\textrm{P}_{loc}(\mathcal{M})$ be the functional  
\begin{equation*}
    F_f[\phi] = \int_{\mathcal{M}}d \mu_x\, f(x) D \phi(x) D^\prime\phi(x),
\end{equation*}
where $f\in\Gamma_0(T\mathcal{M}^{\otimes (i+j)})$. Then, for any $K\in\mathcal{D}^\prime(\mathcal{M}\times\mathcal{M})$, $f,h\in\mathcal{D}(\mathcal{M})$, it holds that, working at the level of integral kernel
\begin{flalign}
    (V_h \star_K F_f)[\phi] &= V_h[\phi] F_f[\phi] - \frac{\hbar^2}{(n-2)!} \int_{\mathcal{M}^2} d\mu_x d\mu_z\, h(x) f(z) D K(x,z) D^\prime K(x,z) \phi^{(n-2)}(x) \notag\\& - \frac{\hbar}{(n-1)!} \int_{\mathcal{M}^2}\,d\mu_x d \mu_z h(x) f(z) [D K(x,z) D^\prime \phi(z) + D \phi(z) D^\prime K(x,z)] \phi^{(n-1)}(x),  \label{Eq: Useful product}
\end{flalign}
where $V_h$ is the interacting potential as per Equation \eqref{Eq: Action Terms}. 
\end{lemma}

\begin{proof}
    The statement comes from a direct application of Equation \eqref{starprod} together with the observation that only the first two terms in the exponential series contribute to Equation \eqref{Eq: Useful product}.
\end{proof}


\noindent Using this lemma we can focus at last on the expectation value of the quantum stress energy tensor. Applying Equation \eqref{Eq: Bogoliubov map}, barring minor adaptations, as in Section \ref{Sec: Examples} and then evaluating at the configuration $\phi=0$ as per Remark \ref{Rem: Expectation value}, we obtain applying recursively Lemma \ref{Lem: Il conto si fa così}
\begin{equation}\label{Eq: Expanded Stress-Energy Tensor}
    R_{\lambda V} (\boldsymbol{\mathcal{T}}^{\eta}_{\mu \, \nu}) \vert_{\phi = 0}(z) = \boldsymbol{\mathcal{T}}^{\eta}_{\mu \, \nu} \vert_{\phi = 0}(z) + R_{\lambda V}^{(2)} (\boldsymbol{\mathcal{T}}^{\eta}_{\mu \, \nu}) \vert_{\phi = 0}(z)+ \mathcal{O}(\lambda^3), 
\end{equation}
where the first order in $\lambda$ vanishes, while $R_{\lambda V}^{(2)} (\textbf{T}_{\mu \, \nu}) \vert_{\phi = 0}(z)$ accounts for the contributions of order $\lambda^2$ and it reads at the level of integral kernel
\begin{flalign}
R_{\lambda V}^{(2)} &(\boldsymbol{\mathcal{T}}^{\eta}_{\mu \, \nu}) \vert_{\phi = 0} (z) = \frac{\lambda^2 \hbar^{(n-1)}}{2 (n-1)!} \int_{\mathcal{M}^2} d \mu_x \, d \mu_y \, \, h(x) h(y)\notag \\ & \left[ 2H^{(n-1)}(x,y) A_{\mu\nu,1}(x,y,z) - H_{AF}^{(n-1)}(x,y) A_{\mu\nu,2}(x,y,z) - H_F^{(n-1)}(x,y)A_{\mu\nu,3}(x,y,z)\right]\label{Eq: Quantum Stress-Energy Tensor Splitting}
\end{flalign}
where 
\begin{flalign*}
A_{\mu\nu,1}(x,y,z)=&\partial_{\mu} H(x,z) \partial_{\nu} H_F(y,z) + \partial_{\nu} H(x,z) \partial_{\mu} H_F(y,z) -g_{\mu \, \nu}(z) \partial_{\rho} H(x,z) \partial^{\rho} H_F(y,z) \\ & - (g_{\mu \, \nu}(z) m^2 - 2 \xi G_{\mu \, \nu}(z)) H(x,z) H_F(y,z) + 4 \xi g_{\mu \, \nu}(z) \partial_{\rho} H(x,z) \partial^{\rho} H_F(y,z) \\& \left. \left. +  2 \xi g_{\mu \, \nu}(z) (H(x,z) \Box H_F(y,z) + \Box H(x,z) H_F(y,z) )- 2 \xi (\partial_{\mu} H(x,z) \partial_{\nu}H_F(y,z) \right. \right. \\ & \left. \left.  + \partial_{\nu} H(x,z) \partial_{\mu} H_F(y,z)) - 2\xi(H(x,z) \nabla_{\mu} \partial_{\nu} H_F(y,z) + \nabla_{\mu} \partial_{\nu} H(x,z) H_F(y,z)) \right. \right. \\ &  + \eta g_{\mu \, \nu}(z) (H(x,z)P_0 H_F(y,z) + P_0H(x,z) H_F(y,z)),
\end{flalign*}
while 
\begin{flalign*}
A_{\mu\nu,2}(x,y,z)=&\partial_{\mu} H(x,z) \partial_{\nu} H(y,z) + \partial_{\nu} H(x,z) \partial_{\mu} H(y,z) -g_{\mu \, \nu}(z) \partial_{\rho} H(x,z) \partial^{\rho} H(y,z) \\ 
& \left. \left. - (g_{\mu \, \nu}(z) m^2 - 2 \xi G_{\mu \, \nu}(z)) H(x,z) H(y,z) + 4 \xi g_{\mu \, \nu}(z) \partial_{\rho} H(x,z) \partial^{\rho} H(y,z) \right. \right. \\
& \left. \left.  + 2 \xi g_{\mu \, \nu}(z) (H(x,z) \Box H(y,z) + \Box H(x,z) H(y,z) ) - 2 \xi (\partial_{\mu} H(x,z) \partial_{\nu}H(y,z) \right. \right. \\ 
& \left. \left.  + \partial_{\nu} H(x,z) \partial_{\mu} H(y,z)) - 2\xi(H(x,z) \nabla_{\mu} \partial_{\nu} H(y,z) + \nabla_{\mu} \partial_{\nu} H(x,z) H(y,z)) \right. \right. \\ 
& \left. \left. + \eta g_{\mu \, \nu}(z) (H(x,z)P_0 H(y,z) + P_0H(x,z) H(y,z)) \right. \right.
\end{flalign*}
and
\begin{flalign*}
A_{\mu\nu,3}(x,y,z)=& \partial_{\mu} H_F(x,z) \partial_{\nu} H_F(y,z) + \partial_{\nu} H_F(x,z) \partial_{\mu} H_F(y,z) - g_{\mu \, \nu}(z) \partial_{\rho} H_F(x,z) \partial^{\rho} H_F(y,z)  \\ 
& \left. \left.  - (g_{\mu \, \nu}(z) m^2 - 2 \xi G_{\mu \, \nu}(z)) H_F(x,z) H_F(y,z) + 4 \xi g_{\mu \, \nu}(z) \partial_{\rho} H_F(x,z) \partial^{\rho} H_F(y,z) \right. \right. \\
& \left. \left. + 2 \xi g_{\mu \, \nu}(z) (H_F(x,z) \Box H_F(y,z) + \Box H_F(x,z) H_F(y,z) ) - 2 \xi (\partial_{\mu} H_F(x,z) \partial_{\nu}H_F(y,z) \right. \right. \\
& \left. \left.  + \partial_{\nu} H_F(x,z) \partial_{\mu} H_F(y,z)) - 2\xi(H_F(x,z) \nabla_{\mu} \partial_{\nu} H_F(y,z) + \nabla_{\mu} \partial_{\nu} H_F(x,z) H_F(y,z)) \right. \right. \\ &  + \eta g_{\mu \, \nu}(z) (H_F(x,z)P_0 H_F(y,z) + P_0H_F(x,z) H_F(y,z)) 
\end{flalign*}

\noindent Observe that in Equation \eqref{Eq: Quantum Stress-Energy Tensor Splitting}, one has to cope with terms proportional to $H^{(n-1)}_F$ and to $H^{(n-1)}_{AF}$ with $n =3,4$. As pointed out in Equation \eqref{Eq: Perturbative Squared Field} and in the subsequent discussion, these are ill-defined as distributions in $\mathcal{M}\times\mathcal{M}$. This can be bypassed, but in a non unique way. Such a renormalization ambiguity cannot be avoided. Here we assume implicitly that we are working with an arbitrary but fixed extension, amounting therefore to the choice of a renormalization scheme.

\begin{remark}\label{Rem: Assumption on H}
    In the following we shall work on a globally hyperbolic spacetime $(\mathcal{M}, g)$ and we replace $H$ with the two-point correlation function of a Gaussian Hadamard state. Yet, with a slight abuse of notation we shall still employ the symbol $H$ as well as $H_F$ for the corresponding Feynman propagator. This fact entails that in Equation \eqref{Eq: Quantum Stress-Energy Tensor Splitting} one can implement additional simplifications, which shall be crucial for the validity of the following results.   
\end{remark}


\begin{theorem}\label{Thm: Conservation Interacting Stress-Energy Tensor}
    Let $(\mathcal{M}, g)$ be the a $4$-dimensional globally hyperbolic spacetime and let $V$ be the interaction potential as in Equation \eqref{Eq: Action Terms}. Let $R_{\lambda V} (\boldsymbol{\mathcal{T}}^{\eta}_{\mu \, \nu})(z)$ be the integral kernel of the quantum, interacting stress-energy tensor as in Equation \eqref{Eq: Extended Intracting Quantum Stress-Energy Tensor}. Then it holds that
    \begin{equation}\label{Eq: Divergence Stress-Energy Tensor order 2}
         \nabla^\mu  R_{\lambda V} (\boldsymbol{\mathcal{T}}^{\eta}_{\mu \, \nu}) \vert_{\phi = 0}(z)=\frac{1-3\eta}{4\pi^2}\partial_\nu v_1(z,z)+\mathcal{O}(\lambda^3),
    \end{equation}
    if and only if $\eta = \frac{1}{n}$. 
\end{theorem}
\begin{proof}
    Focusing first of all on $R_{\lambda V}^{(2)} (\boldsymbol{\mathcal{T}}^{\eta}_{\mu \, \nu}) \vert_{\phi = 0}(z)$, in Equation \eqref{Eq: Quantum Stress-Energy Tensor Splitting}, we have identified three terms proportional respectively to $H^{(n-1)}(x,y)$, $H^{(n-1)}_{AF}(x,y)$ and $H^{(n-1)}_F(x,y)$. From an algebraic viewpoint their contribution is the same namely of the form $H^{(n-1)}(x,y)A_{\mu\nu}(x,y,z)$ where we suppress the subscript which distinguished between the possible scenarios. Hence, by applying Leibnitz rule, 
    $$\nabla^\mu H^{(n-1)}(x,y)A_{\mu\nu}(x,y,z)=(n-1) H^{(n-2)}\nabla^\mu H(x,y) A_{\mu\nu}(x,y,z)+H^{(n-1)}(x,y) \nabla^\mu A_{\mu\nu}(x,y,z).$$
In turn, in all cases under scrutiny, the kernel $A_{\mu\nu}(x,y,z)$ has the same structure and we denote it for simplicity as 
    \begin{flalign*}
A_{\mu\nu}(x,y,z)=&\partial_{\mu} H_1(x,z) \partial_{\nu} H_2(y,z) + \partial_{\nu} H_1(x,z) \partial_{\mu} H_2(y,z) -g_{\mu \, \nu}(z) \partial_{\rho} H_1(x,z) \partial^{\rho} H_2(y,z) \\ & - (g_{\mu \, \nu}(z) m^2 - 2 \xi G_{\mu \, \nu}(z)) H_1(x,z) H_2(y,z) + 4 \xi g_{\mu \, \nu}(z) \partial_{\rho} H_1(x,z) \partial^{\rho} H_2(y,z) \\& \left. \left. +  2 \xi g_{\mu \, \nu}(z) (H_1(x,z) \Box H_2(y,z) + \Box H_1(x,z) H_2(y,z) )- 2 \xi (\partial_{\mu} H_1(x,z) \partial_{\nu}H_2(y,z) \right. \right. \\ & \left. \left.  + \partial_{\nu} H_1(x,z) \partial_{\mu} H_2(y,z)) - 2\xi(H_1(x,z) \nabla_{\mu} \partial_{\nu} H_2(y,z) + \nabla_{\mu} \partial_{\nu} H_1(x,z) H_2(y,z)) \right. \right. \\ &  + \eta g_{\mu \, \nu}(z) (H_1(x,z)P_0 H_2(y,z) + P_0H_1(x,z) H_2(y,z)).
\end{flalign*}
This entails that $\nabla^{\mu} A_{\mu \, \nu}(x,y,z)$ reads
    \begin{flalign}
        \nabla^{\mu} &\left[\partial_{\mu} H_1(x,z) \partial_{\nu} H_2(y,z) + \partial_{\nu} H_1(x,z) \partial_{\mu} H_2(y,z) \right. \notag\\ 
        & \left. -g_{\mu \, \nu}(z) \partial_{\rho} H_1(x,z) \partial^{\rho} H_2(y,z) - (g_{\mu \, \nu}(z) m^2 - 2 \xi G_{\mu \, \nu}(z)) H_1(x,z) H_2(y,z) \right.  \notag\\
        &  \left. + 4 \xi g_{\mu \, \nu}(z) \partial_{\rho} H_1(x,z) \partial^{\rho} H_2(y,z) + 2 \xi g_{\mu \, \nu}(z) (H_1(x,z) \Box H_2(y,z) \right. \notag\\ 
        & \left. \Box H_1(x,z) H_2(y,z) ) - 2 \xi (\partial_{\mu} H_1(x,z) \partial_{\nu}H_2(y,z) + \partial_{\nu} H_1(x,z) \partial_{\mu} H_2(y,z))  \right. \notag\\ 
        & \left. - 2\xi(H_1(x,z) \nabla_{\mu} \partial_{\nu} H_2(y,z) + \nabla_{\mu} \partial_{\nu} H_1(x,z) H_2(y,z)) \right] = \notag \\ 
        &= \Box H_1(x,z) \partial_{\nu} H_2(y,z) + \partial_{\nu} H_1 (x,z) \Box H_2(y,z) \notag \\ 
        &-(m^2 + \xi R) (\partial_{\nu} H_1(x,z) H_2(y,z) + H_1(x,z) \partial_{\nu} H_2(y,z)) \notag\\ 
        & 2\xi (H_1(x,z) \partial_{\nu} \Box H_2(y,z) - H_1(x,z) \Box \partial_{\nu} H_2(y,z) + R_{\mu \, \nu} H_1(x,z) \partial^{\mu} H_2(y,z)) \notag
        \\ & + 2 \xi (\partial_{\nu} \Box H_1(x,z) H_2(y,z) - \Box \partial_{\nu} H_1(x,z) H_2(y,z) + R_{\mu \, \nu} \partial^{\mu} H_1(x,z) H_2(y,z)), \label{Eq: Usueful Formula}
    \end{flalign}
where $H_1,H_2$ are suitable parametrices depending on the various cases. Using the identity valid on scalar functions
\begin{equation}
    \nabla_{\nu} \Box  - \Box \nabla_{\nu}  = - R_{\mu \, \nu} \nabla^{\mu},
\end{equation}
the last two terms in Equation \eqref{Eq: Usueful Formula} vanish and we obtain 
\begin{equation*}
    \nabla^{\mu} [...] = - P_0 H_1(x,z) \partial_{\nu} H_2(y,z) - \partial_{\nu} H_1(x,z) P_0 H_2(y,z). 
\end{equation*}
Summing up, we have obtained, at the level of integral kernel
\begin{flalign}
\label{Eq: nabla R}
    \notag \nabla^{\mu} &R_{\lambda V} (\boldsymbol{\mathcal{T}}^{\eta}_{\mu \, \nu}) \vert_{\phi = 0} (z) = \frac{\lambda^2 \hbar^{(n-1)}}{2(n-1)!}\int_{\mathcal{M}^2} d\mu_x d\mu_y h(x) h(y) \\ \notag
    &\left \{ 2 H^{(n-1)}(x,y) [-P_0 H(x,z) \partial_{\nu} H_F(y,z) - \partial_{\nu} H(x,z) P_0 H_F(y,z)]  \right. \\ \notag
    & \left. -  H_{AF}^{(n-1)}(x,y) [-P_0 H(x,z) \partial_{\nu} H(y,z) - \partial_{\nu} H(x,z) P_0 H(y,z)] \right. \\ \notag
    & \left. -  H_{F}^{(n-1)}(x,y) [-P_0 H_F(x,z) \partial_{\nu} H_F(y,z) - \partial_{\nu} H_F(x,z) P_0 H_F(y,z)] \right\}  \\
    & + \eta \frac{\lambda^2 \hbar^{(n-1)}}{(n-1)!} \int_{\mathcal{M}} d\mu_y \,  h(y) g_{\mu \, \nu}(z) \nabla^{\mu} [H^{n}(z,y) - H_F^n(z,y)]+\mathcal{O}(\lambda^3).
\end{flalign}\\
Focusing on the last two terms in the previous equation, it holds that
\begin{flalign*}
    \int_{\mathcal{M}} &d\mu_y \, h(y)\,  g_{\mu \, \nu}(z) \nabla^{\mu} [H^{n}(z,y) - H_F^n(z,y)] = \\ & = n \int_{\mathcal{M}} d\mu_y\,  h(y)  [H^{(n-1)}(z,y) \partial_{\nu} H(z,y) - H_F^{(n-1)}(z,y) \partial_{\nu} H_F(z,y)]. 
\end{flalign*}

\noindent Bearing in mind Remark \ref{Rem: Assumption on H}, we can further simplify Equation \eqref{Eq: nabla R} exploiting the following relations, which hold true at the level of integral kernels, \textit{i.e.},  
\begin{equation*}
    P_0 H(x,z) = 0, \, \; \; \, P_0 H_F(y,z) = \delta(y,z). 
\end{equation*}

\noindent Hence, it is possible to infer that

\begin{flalign*}
    \nabla^{\mu} &R_{\lambda V}^{(2)} (\boldsymbol{\mathcal{T}}^{\eta}_{\mu \, \nu}) \vert_{\phi = 0}(z) = \frac{\lambda^2 \hbar^{(n-1)}}{2(n-1)!}\int_{\mathcal{M}} d\mu_y \, h(y) \, \\ &\left \{ -2 H^{(n-1)}(z,y) \partial_{\nu} H(z,y)] +2 H_{F}^{(n-1)}(z,y) \partial_{\nu} H_F(z,y) \right\} \\ & + \eta \frac{\lambda^2 \hbar^{(n-1)} n}{(n-1)!} \int_{\mathcal{M}} d\mu_y \,  h(y)   [H^{(n-1)}(z,y) \partial_{\nu} H(z,y) - H_F^{(n-1)}(z,y) \partial_{\nu} H_F(z,y)]+\mathcal{O}(\lambda^3).
\end{flalign*}
As a consequence, if we set $\eta = \frac{1}{n}$ either for $n=3$ or $n=4$, the last expression vanishes. To conclude the proof we need to control also the leading order contribution, namely $\nabla^\mu (\boldsymbol{\mathcal{T}}^{\frac{1}{n}}_{\mu \, \nu}) \vert_{\phi = 0}(z)$. We can rely on \cite[Thm 2.1]{moretti} and on its proof, where it is shown that, on a $4$-dimensional spacetime
\begin{equation}\label{Eq: epiphany}
\nabla^\mu  R_{\lambda V} (\boldsymbol{\mathcal{T}}^{\eta}_{\mu \, \nu}) \vert_{\phi = 0}(z)=\frac{1-3\eta}{4\pi^2}\partial_\nu v_1(z,z),
\end{equation}
where $v_1(x,y)$ is the coefficient appearing in the expansion in Equation \eqref{Eq: Hadamard expansion}. 
\end{proof}

\begin{corollary}\label{Cor: Conservation on Minkowski}
    Under the same hypothesis of Theorem \ref{Thm: Conservation Interacting Stress-Energy Tensor} and assuming in addition that $(\mathcal{M}, g)$ coincides with the four-dimensional Minkowski spacetime, then 
    $$\nabla^\mu   R_{\lambda V} (\boldsymbol{\mathcal{T}}^{\eta}_{\mu \, \nu}) \vert_{\phi = 0}(z)=\mathcal{O}(\lambda^3).$$
\end{corollary}

\begin{proof}
    The statement descends from Equation \eqref{Eq: Divergence Stress-Energy Tensor order 2} observing that, on account of Equation \eqref{Eq: v1}, $v_1(z,z)=\frac{m^4}{8}$, hence it is constant.
\end{proof}

Having established an expression for the divergence of the quantum stress-energy tensor up to second order in perturbation theory, we can investigate the associated trace. In the following we focus once more on a four-dimensional background and we assume that the field is conformally coupled to the scalar curvature and that the interaction potential is quartic in the field.

\begin{theorem}\label{Thm: Interacting Trace Anomaly}
    Let $R_{\lambda V} (\boldsymbol{\mathcal{T}}^{\eta}_{\mu \, \nu})(x) \vert_{\phi = 0}$ be the integral kernel of the interacting, quantum stress-energy tensor expanded up to second order in perturbation theory, as per Equation \eqref{Eq: Expanded Stress-Energy Tensor} on a globally hyperbolic spacetime $(\mathcal{M},g)$. Then, if $\eta = \frac{1}{4}$ and $\xi = \frac{1}{6}$, it holds that 
    \begin{equation}\label{eq: Trace Anomaly 2nd order}
        g^{\mu \, \nu} R_{\lambda V} (\boldsymbol{\mathcal{T}}^{\eta}_{\mu \, \nu})(z) \vert_{\phi = 0} =  - m^2 (R_{\lambda V}^{(0)} (\Phi^2)+R^{(2)}_{\lambda V}(\Phi^2))\vert_{\phi = 0}(z)+\mathcal{O}(\lambda^3).
    \end{equation}
\end{theorem}

\begin{proof}
    Observe that, starting from Equation \eqref{Eq: Expanded Stress-Energy Tensor}, two contributions appear in the trace. The first can be read as as a modification of Equation \eqref{traceT} to the case $\eta = \frac{1}{4}$. Hence, using Equation (72) in \cite{moretti}, it holds that 
    \begin{equation*}
         \boldsymbol{\mathcal{T}}^{\frac{1}{4}}_{\mu \, \nu}(z) \vert_{\phi = 0} = - m^2 \Phi^2(z) \vert_{\phi = 0},
    \end{equation*}\
    where $\Phi^2(z) \vert_{\phi = 0} = R_{\lambda V}^{(0)} (\Phi^2)$. Observe that, contrary to the case with $\eta=\frac{1}{3}$ in the free field theory, namely Equation \eqref{traceT}, there is no contribution coming from $v_1(z,z)$ since, in Equation (72) in \cite{moretti}, this is proportional to a factor $(4\eta-1)$. The second contribution is given by the trace of $R_{\lambda V}^{(2)}(\boldsymbol{\mathcal{T}}^{\eta}_{\mu \, \nu})(z)$. In the same spirit of the proof of Theorem \ref{Thm: Conservation Interacting Stress-Energy Tensor}, we exploit Equation \eqref{Eq: Quantum Stress-Energy Tensor Splitting}. Hence we need to consider three terms which have the same algebraic structure and, using the same notation as in Equation \eqref{Eq: Usueful Formula}, we can write at the level of integral kernels 
    \begin{flalign*}
        g^{\mu \, \nu} A_{\mu\nu}(x,y,z) =& \frac{\lambda^2 \hbar^3}{12}\left[ 2 \partial^{\nu} H_1(x,z) \partial_{\nu}H_2(y,z) + \right.  - 4 \partial_{\nu} H_1(x,z) \partial^{\nu} H_2(y,z) - 4 m^2H_1(x,z) H_2(y,z) \\ 
        & - 2 \xi R(z) H_1(x,z) H_2(y,z) +  16 \xi \partial_{\nu} H_1(x,z) \partial^{\nu} H_2(y,z) + 8 \xi (H_1(x,z) \Box H_2(y,z) \\
        & \left.+ \Box H_1(x,z) H_2(y,z)) - 4 \xi \partial^{\nu} H_1(x,z) \partial_{\nu} H_2(y,z) - 2\xi (H_1(x,z) \Box H_2(y,z) \right. \\
        &  \left. + \Box H_1(x,z) H_2(y,z)) + (H_1(x,z) [-\Box + m^2 + \xi R] H_2(y,z)) - P_0 H_1(x,z) H_2(y,z)\right. \\
        & = \left[ 2 (6 \xi -1) \partial_{\nu} H_1(x,z) \partial^{\nu} H_2(y,z) - 2 m^2 H_1(x,z) H_2(y,z) \right. \\
        & + \left. (6\xi-1) H_1(x,z) \Box H_2(y,z) + (6\xi -1) \Box H_1(x,z) H_2(y,z) \right],
    \end{flalign*}
    where all derivatives are understood to be with respect to the $z$ variable. In the above equation, we exploited the fact that \begin{equation*}
    \nabla^{\mu} \partial_{\mu} H_i(x,z) = \Box H_i(x,z), \, i = 1,2. 
\end{equation*} and we made use of the relations
\begin{equation*}
    g^{\mu \, \nu} g_{\mu \, \nu} = 4, \, \; \; \; \, g^{\mu \, \nu} G_{\mu \, \nu} = -R,
\end{equation*} together with Equation \eqref{Eq: KG}. Recall that, as in Theorem \ref{Thm: Conservation Interacting Stress-Energy Tensor}, see also Remark \ref{Rem: Assumption on H}, we are assuming that we are working with the two-point correlation function of a Hadamard state and with its associated Feynman propagator. Hence, if we set $\xi = \frac{1}{6}$ and if we gather all results, we end up with 
   \begin{flalign*}
        g^{\mu \, \nu} & R_{\lambda V}^{(2)} (\boldsymbol{\mathcal{T}}^{\eta}_{\mu \, \nu})(z) \vert_{\phi = 0} = - 2 m^2 \frac{\lambda^2 \hbar^3}{12} \int_{\mathbb{M}^2} d\mu_x \, d \mu_y \, h(x) h(y) \\ & \left\{ 2 H^3(x,y) H(x,z) H_F(y,z) - H^3_{AF}(x,y) H(x,z) H(y,z) - H^3_F (x,y) H_F(x,z) H_F(y,z) \right\} = \\ & = - m^2 R^{(2)}_{\lambda V} (\Phi^2)(z)\vert_{\phi = 0}, 
    \end{flalign*}
which is the sought result.
\end{proof}

Theorem \ref{Thm: Interacting Trace Anomaly} combined with Corollary \ref{Cor: Conservation on Minkowski} shows that, for a self-interacting real scalar field with a quartic potential and at second order in perturbation theory, it is possible to construct on Minkowski spacetime a conserved quantum stress-energy tensor, but the price to pay is a modification of the underlying quantum trace. This is no longer stemming from an anomalous term as in the free theory but it is built out of a contribution proportional to the squared, interacting field, which vanishes in the massless case. This is a by-product of the modification of the coefficient $\eta$ from $\frac{1}{3}$ in the free field theory, see \cite{moretti}, to $\eta=\frac{1}{4}$ as in Theorem \ref{Thm: Conservation Interacting Stress-Energy Tensor}, which makes the potential quadratic in the underlying fields.

 \begin{remark}
    If we consider on Minkowski spacetime a cubic self-interaction, it is possible to repeat mutatis mutandis the same analysis as in Theorem \ref{Thm: Interacting Trace Anomaly}. We limit ourselves at reporting the counterpart of Equation \eqref{eq: Trace Anomaly 2nd order}, namely, setting $\eta=\frac{1}{3}$ and $\xi=\frac{1}{6}$, we obtain at the level of integral kernels
    \begin{gather*}
 g^{\mu \, \nu} R_{\lambda V} (\boldsymbol{\mathcal{T}}^{\frac{1}{3}}_{\mu \, \nu})(z) \vert_{\phi = 0} =\\
    \frac{1}{4\pi^2}v_1(z,z)+\frac{\lambda^2 \hbar^2}{6} \int_{\mathcal{M}^2} d\mu_x \, [H^3(x,z) - H^3_F(x,z)] h(x) h(z)  - \frac{5}{3} m^2 \left.R^{(2)}_{\lambda V}(\Phi^2)(z)\right|_{\phi=0},
    \end{gather*}
    where $v_1(z,z)=\frac{m^4}{8}$, see Equation \eqref{Eq: v1}. The first term coincides with the anomalous contribution present in the free field theory in \cite{moretti} since, in this scenario,  $\eta=\frac{1}{3}$, precisely as in Theorem \ref{valter}, while the second term, proportional to $\lambda^2$, can be ascribed to the effect of the interaction potential. 
\end{remark}

\begin{remark}
    We observe that our result is in full agreement with \cite{Hathrell:1981zb}. Herein, using path-integral techniques, it is argued that, if one considers the conformally invariant case, that is $m=0$ and $\xi=\frac{1}{6}$, the first $\lambda$-dependent contribution to the trace anomaly should appear at the next order in perturbation theory. We plan to confirm this statement in a future work. It is worth mentioning that setting $\eta=\frac{1}{n}$ entails that we are considering a quantum stress-energy tensor which is quadratic in the fields. This supports the statement that, setting $m=0$, we obtain a state independent trace anomaly. However, if, at higher orders in perturbation theory, $\eta$ becomes $\lambda$-dependent, this feature would no longer be present.
\end{remark}

\begin{remark}
 One might wonder if the results obtained in Corollary \ref{Cor: Conservation on Minkowski} can be extended to a generic, four-dimensional, globally hyperbolic spacetime. Equation \eqref{Eq: epiphany} entails that one needs to address the issue of $\partial_\mu v_1(z,z)$. A direct comparison with Equation \eqref{Eq: v1} entails that, unless one considers rather special backgrounds, such as all maximally symmetric spacetimes, this is not necessarily vanishing.  Yet, as highlighted in \cite{hollands2}, one can exploit that, already at the free field level, there exists a regularization ambiguity in defining the Wick ordered stress-energy tensor. This can be used to cancel the non vanishing contribution in Equation \eqref{Eq: epiphany} adding a geometric, hence local and covariant tensor of the form
    $$Q_{\mu\nu}=\frac{1}{\pi^2}g_{\mu\nu}v_1(z,z),$$
    where we have replaced the value $\eta=\frac{1}{4}$ considering a quartic interaction. In other words this seems to point into the direction that, contrary to what was expected originally in the literature, it is possible to extend the approach of \cite{moretti} also to work for interacting field models, at least at the level of perturbation theory, though this requires to account for a larger set of renormalization freedoms as pointed in \cite{hollands2}. It is worth mentioning that, in this scenario, Theorem \ref{Thm: Interacting Trace Anomaly}, Equation \eqref{eq: Trace Anomaly 2nd order} in particular, would need to be modified only by adding a state independent term:
    $$Q=g^{\mu\nu}Q_{\mu\nu}=\frac{4}{\pi^2}v_1(z,z).$$
    contribution due to $Q$. This can be ascribed to the fact that the value of $\eta$ does not depend on the choice of the underlying background and, setting $\eta=\frac{1}{4}$ cancels any additional anomalous contribution.
\end{remark}

\section*{Acknowledgements}
B.C. is supported by a PhD fellowship of the University of Pavia whose support is gratefully acknowledged. C.D. acknowledges the support of the Gruppo Nazionale di Fisica Matematica. This work is based partly on the thesis of M.G., titled {\em On the Stress-Energy Tensor of a Self-Interacting Quantum Scalar Field in the pAQFT approach} submitted to the University of Pavia in partial fulfilment of the requirements for completing the master degree program in physics. In addition B.C. and C.D. are grateful to the COST Action CA23115 - Relativistic Quantum Information. We are grateful to Markus Fr\"ob for the useful discussion and for pointing out in particular reference \cite{Hathrell:1981zb}.

\vskip.2cm

\noindent\textbf{Data availability statement}. Data sharing is not applicable to this article as no new data were created or analysed in this study.

\vskip .2cm

\noindent\textbf{Conflict of interest statement.} The authors certify that they have no affiliations with or involvement in any
organization or entity with any financial interest or non-financial interest in the subject matter discussed in
this manuscript.

\addcontentsline{toc}{section}{Acknowledgements}

\begin{appendices}

\section{Microlocal analysis: Wavefront set and scaling degree}
\label{First appendix}
This appendix is devoted to the illustration of some basic notions and results in \emph{microlocal analysis}, which plays a crucial r\^ole in the rest of the paper. The following discussion is mainly based on \cite[Ch. 8]{horm1}, although specialised to the case of open subsets of $\mathbb{R}^n$. This choice is motivated by the fact that each real, smooth, $n-$dimensional manifold $M$ can be locally modelled as an open subset of $\mathbb{R}^n$. Thus, the analysis carried on in this appendix supplies the reader with all the key tools necessary for an extension to a more general scenario. We shall omit the proofs of the various results since they fall outside the scopes of this work and they can be easily found in the literature -- see \cite[Ch.1]{bar} and \cite[Chap. 6]{horm1}. \\
Microlocal analysis is concerned with the study of the singular structure of distributions, which naturally emerge in the context of quantum field theories. To wit, in the present work, we are interested in the characterisation of the singularities of the Hadamard parametrices which can be dealt with perturbative AQFT. To this avail, we shall first set the notation used throughout this work. \\
Let $U \subseteq \mathbb{R}^n$ be an open set and fix a point $x \in U$. We denote by $T_x U$ and $T^*_xU$ the tangent and cotangent space at $x$, respectively. Moreover, let $\mathring{T}_x U := T_x U \setminus \{0\}$ be \emph{punctured} tangent space at $x$ and, similarly, denote by $\mathring{T}^*_x U : = T_x^* U \setminus \{0\}$ the \emph{punctured} cotangent spaces at $x$. We define the tangent and cotangent bundles over $U$ as 
\begin{flalign*}
    TU &:= \bigsqcup_{x \in U} T_x U, \\
    T^*U &:= \bigsqcup_{x \in U} T^*_x U. 
\end{flalign*}
Smooth sections of $TU$ are called \emph{vector fields} and they shall be denoted by $\Gamma(TU)$, whilst smooth sections of $T^*U$ are named \emph{covector fields} and denoted by $\Gamma(T^*U)$. 
Furthermore, let $\mathcal{D}(U)$ be the space of test functions, \textit{i.e.}, of smooth, compactly supported functions on $U \subseteq \mathbb{R}^n$, endowed with the canonical LF-topology. 
The topological dual space of $\mathcal{D}(U)$,  denoted by $\mathcal{D}'(U)$, is the space of distributions on $U$. Similarly, we call $\mathcal{E}'(U)$ the space of distributions with compact support on $U$, that is the dual space with respect to the strong topology to the space of smooth test functions $\mathcal{E}(U)$. \\
With these premises, we can give a preliminary geometric notion necessary for the ensuing analysis, namely the definition of \emph{conic neighbourhood}. 
\begin{definition}[Conic neighbourhood]
\label{connei}
    Let $U \subseteq \mathbb{R}^n$ be open and let $x \in U$. A subset $\Gamma \subseteq \mathring{T}^*_x U$ is said to be \emph{conic} if $\forall \xi \in \Gamma$ and $t > 0$, $t\xi \in \Gamma$. Furthermore, such a subset $\Gamma$ is called \emph{conic neighbourhood} of $\xi \in \mathring{T}^*_x U$ if $\xi \in \Gamma$ and if it is a conic subset of $\mathring{T}^*_x U$.
\end{definition}
\noindent The first step towards the characterisation of the singular structure of a distribution is the study of its \emph{singular support}, \textit{i.e.}, the set of all points of $\mathbb{R}^n$ at which a given distribution fails to be locally approximated by a smooth function. To introduce this concept, we first need to define the notion of \emph{regular direction}. 
\begin{definition}[Regular Direction]
    Let $U \subseteq \mathbb{R}^n$ and let $u \in \mathcal{D}'(U)$, where $\mathcal{D}'(U)$ denotes the space of distributions over $U$. A pair $(x, \xi) \in \mathring{T}^* U$ is said to be a \emph{regular direction} of $u$ if there exist $\psi \in C^{\infty}_0(U), \psi \ne 0$, a conic neighbourhood $\Gamma$ of $\xi$ as per Definition \ref{connei}, and constants $C_n \in \mathbb{R}$, $\forall n \in \mathbb{N}$ such that 
    \begin{equation*}
        |\widehat{\psi u} (\xi)| < \frac{C_n}{1 + |\xi|^n}, \, \, \,  \forall \xi \in \Gamma, \forall n \in \mathbb{N},
    \end{equation*}
    where $\widehat{\cdot}$ denotes the Fourier transform, whilst $|\cdot|$ denotes the Euclidean norm on $\mathbb{R}^n$. We say that $\widehat{\psi u} (\xi)$ is rapidly decreasing along the direction $\xi \in \Gamma$. Moreover, the set 
    \begin{equation*}
        \Sigma(u) := \mathring{\mathbb{R}}^n \setminus \left\{ \xi \in \mathring{\mathbb{R}}^n \vert \, u \, \text{is regular in the direction of} \, \xi \right\}, 
    \end{equation*}
    is said to be the \emph{frequency set} of $u$. 
\end{definition}
\noindent At this stage, we can give one of the key definitions of this appendix. 
\begin{definition}[Singular support]
Let $U$ be an open subset of $\mathbb{R}^n$ and let $v \in \mathcal{D}'(U)$ be a distribution. The \emph{singular support} of $v$ is the set 
\begin{equation*}
    singsupp(v) := U \setminus \{ x \in U \, | \, \exists \omega \in \mathcal{N}_x, \exists \phi \in C^{\infty}(\omega) : \langle v, \phi \rangle = \langle u_{\phi}, f \rangle, \forall f \in \mathcal{D}(\omega) \}, 
\end{equation*}
where $\mathcal{N}_x$ denotes the set of all open neighbourhoods of $x \in U$, while $u_{\phi}$ is the distribution generated by $\phi$ via integration against a test function. 
\end{definition}
\begin{definition}[Singular Cone]
    Let $U \subseteq \mathbb{R}^n$ be open and let $u \in \mathcal{D}'(U)$. For $x \in U$, we define 
    \begin{equation*}
        \mathscr{S}_x(U) := \{g \in C^{\infty}_0(\mathbb{R}^n) \, | \, supp(g) \subset U, \, g(x) \ne 0\}
    \end{equation*}
    the collection of \emph{probes at $x$}. We call \emph{singular cone} of $u$ at $x$ the set
    \begin{equation*}
        \Sigma_x(u) := \bigcap_{g \in \mathscr{S}_x(U)} \Sigma (gu) \subseteq \mathring{\mathbb{R}}^n. 
    \end{equation*}
\end{definition}
\noindent The above definition introduces an auxiliary concept, necessary to present the notion of wavefront set of a distribution. 
\begin{definition}[Wavefront set]
    Let $U \subseteq \mathbb{R}^n$ be open and let $u \in \mathcal{D}'(U)$. The set 
    \begin{equation*}
        WF(u) := \bigsqcup_{x \in U} \Sigma_x(u) = \left\{ (x,\xi) \in \mathring{T}^* U \, | \, \xi \in \Sigma_x(u) \right\}, 
    \end{equation*}
    is called \emph{wavefront set} of $u$. 
\end{definition}
\begin{example}
Consider a real scalar field on a globally hyperbolic spacetime $(\mathcal{M},g)$. If we denote by $\Delta \in \mathcal{D}'(\mathcal{M} \times \mathcal{M})$ the fundamental kernel associated to the causal propagator of the theory via Schwartz kernel theorem -- see \cite[Sec. 5.2]{horm1} --, it holds that 
\begin{equation*}
    WF(\Delta) := \left \{(x,y, \xi_x, - \xi_y) \in \mathring{T}^*(\mathcal{M} \times \mathcal{M}) \, \vert \, (x,\xi_x) \sim (y,\xi_y) \right\}, 
\end{equation*}
where the symbol $\sim$ means that there exists a null geodesic $\gamma$ connecting $x$ and $y$ such that $\xi_x$ is parallel and tangent to $\gamma$ at $x$ and $\xi_y$ is the parallel transport from $x$ to $y$ along $\gamma$ of $\xi_x$. 
\end{example}
\noindent In the following, we shall state some useful properties of the wavefront set of a distribution, starting from \emph{H\"ormander's criterion}, which establishes a sufficient condition for the well-posedness of the product of two distributions. 
\begin{theorem}[H\"ormander's criterion]\label{Thm: Hormander criterion}
Let $U \subseteq \mathbb{R}^n$ be open and let $u,v \in \mathcal{D}'(U)$ be such that for all $(x, \xi) \in WF(u)$, $(x, -\xi) \notin WF(v)$. Then, there exists the pointwise product $uv \in \mathcal{D}'(U)$, defined as 
\begin{equation*}
    uv := \iota^*(u \otimes v),
\end{equation*}
where $\iota^*$ denotes the pullback along the diagonal of $U \times U$. \\
Furthermore, the wavefront set of the product distribution can be gauged with the condition
\begin{equation*}
    WF(uv) \subseteq \{(x, \xi) \in \mathring{T}^*U \, | \, \xi = \eta + \zeta \ne 0 \}, 
\end{equation*}
where $(x, \eta) \in WF(u) \cup (x,0)$, $(x, \zeta) \in WF(v) \cup (x,0)$. 
\end{theorem}
\begin{proof}
    \noindent The proof of this theorem can be found in \cite[Thm. 8.2.10]{horm1}. 
\end{proof}
Other relevant properties of the wavefront set are listed in the following theorem.
\begin{theorem}
    Let $U \subseteq \mathbb{R}^n$ be an open subset. Then, 
    \begin{itemize}
        \item[1.] If $P$ is a differential operator and $u \in \mathcal{D}'(U)$, 
        \begin{equation}
            WF(Pu) \subseteq WF(u) \subseteq WF(Pu) \cup Char(P),
        \end{equation}
        where $Char(P)$ is the characteristic set of $P$, \textit{i.e.},  
        \begin{equation*}
            Char(P) := \{(x, \xi) \in \mathring{T}^*U \, | \, \sigma_P(x, \xi) = 0 \}, 
        \end{equation*}
        $\sigma_P(x, \xi)$ being the principal symbol of $P$.
    \item[2.] Consider an open subset $V \subseteq \mathbb{R}^k$, where $k$ is not necessary equal to $n$, and let $u \in \mathcal{D}'(U)$, $v \in \mathcal{D}'(V)$, then the wavefront set of the tensor product $u \otimes v$ is such that 
    \begin{flalign}
        WF(u \otimes v) \subset &(WF(u) \times WF(v)) \cup ((supp(u) \times \{0\}) \times WF(v)) \\ \notag & \cup (WF(u) \times (supp(v) \times \{0\})). 
    \end{flalign}
    \item[3.] Let $V \subseteq \mathbb{R}^k$ be open, let $K \in \mathcal{D}'(U \times V)$ and $u \in \mathcal{E}'(U)$. We set 
    \begin{equation*}
        WF'_2(K) := \{(x_2, \xi_2) \in T^*V \, \vert \, \exists x_1 \in U, (x_1, x_2, 0, -\xi_2) \in WF(K)\}. 
    \end{equation*}
    If $WF'_2(K) \cap WF(u) = \emptyset$, then $K \circledast u \in \mathcal{D}'(U)$, where $(K \circledast u) (f) := K(f \otimes u)$, for all $f \in \mathcal{D}(U)$. In addition, 
    \begin{equation}
        WF(K \circledast u) \subseteq WF_1(K) \cup WF'(K) \circ WF(u), 
    \end{equation}
    where we have set 
    \begin{flalign*}
        WF_1(K) &:= \{(x_1, \xi_1) \in \mathring{T}^* U \, \vert \, \exists x_2 \in V, (x_1, x_2, \xi_1, 0) \in WF(K)\}, \\
        WF'(K) &:= \{ (x_1, x_2, \xi_1, \xi_2) \in \mathring{T}^*(U \times U) \, \vert \, (x_1, x_2, \xi_1, -\xi_2) \in WF(K)\}, \\
        WF'(K) \circ WF(u) &:= \{(x_1, \xi_1) \in \mathring{T}^* U \, \vert \, \exists (x_2, \xi_2) \in WF(u), (x_1, x_2, \xi_1, - \xi_2) \in WF(K) \}.\\
    \end{flalign*}
    \end{itemize}
\end{theorem}
\begin{proof}
    The proof of this theorem can be found in \cite[Thms. 8.2.9, 8.2.12, 8.2.13, 8.2.14]{horm1}. 
\end{proof}
Despite it has proven essential in the identification of singularities, the notion of wavefront set does not suffice for our aims. Indeed, the wavefront set does not distinguish between the singular structure of different distributions, such as the Dirac delta distribution and the one generated by the continuous function $f(x) :=|x|$. Therefore, we need to introduce a more refined and computationally efficient tool, \emph{i.e.}, the \emph{(Steinman) scaling degree}, which will provide a powerful criterion to determine the existence of admissible extensions of singular distributions. The following exposition is mainly modeled on \cite{brunetti}, which translates Steinman's results in a language more akin to the microlocal formalism. \\
As a first step, let us discuss how to extend a distribution which is ill-defined at a point. In the following, fixed a point $y \in \mathbb{R}^n$ and a conic open subset $U \subseteq \mathbb{R}^n$, we denote by $U_y := U + y$. Furthermore, for all $f \in \mathcal{D}(U)$ and for any $\lambda > 0$, we set $f^{\lambda}_{y} := \lambda^{-n} f(\lambda^{-1} (x-y)) \in \mathcal{D}(U_y)$. 
\begin{definition}[Scaling degree at a point]
\label{scaldeg}
In the aforementioned scenario, given $u \in \mathcal{D}'(U)$, we denote as $u^{\lambda}_y \in \mathcal{D}'(U_y)$ the distribution defined by $u^{\lambda}_y(f) := u(f^{\lambda}_y)$, for all $f \in \mathcal{D}(U)$. The \emph{scaling degree} of $u$ at $y$ is given by
\begin{equation}
    sd_y(u) := \inf \left\{ \omega \in \mathbb{R} \, \vert \, \lim_{\lambda \rightarrow 0} \lambda^{\omega} u^{\lambda}_y = 0 \right\}. 
\end{equation}
\end{definition}
\noindent Before discussing how this notion can be generalised to an embedded submanifold, let us present a notable example.
\begin{example}
    Consider the Dirac delta distribution centered at $y \in \mathbb{R}^n$, \textit{i.e.}, $\delta_y \in \mathcal{D}'(\mathbb{R}^n)$. A straightforward application of Definition \ref{scaldeg} shows that $\delta^{\lambda}_y = \lambda^{-n} \delta_y$ and, hence, $sd(\delta_{y}) = n$. 
\end{example}
\noindent We can now extend the previous definition to the case in which a distribution is ill-defined at an embedded submanifold. To this avail, we shall need a bunch of preliminary notions. 
\begin{definition}[Admissible]
    Let $M$ be a smooth, $n-$dimensional manifold and let $N \subset M$ be an embedded submanifold via the embedding map $\iota: N \hookrightarrow M$. Let $V$ be a star-shaped neighborhood of the zero section $Z(T_NM)$ of $T_N M := \bigcup_{x \in N} T_x M$. We call \emph{admissible} a map
    \begin{equation*}
        \alpha: V \rightarrow N \times M,
    \end{equation*}
    which is a diffeomorphism onto its range and satisfies the following properties: 
    \begin{itemize}
        \item[1.] $\alpha (x,0) = (x,x), x \in N$,
        \item[2.] $\alpha(TN \cap V) \subset N \times M$,
        \item[3.] $\alpha(x,v) \in \{x\} \times M, x \in N$ and $v \in T_xM$, 
        \item[4.] $d_{\nu} \alpha (x, \cdot) \vert_{v = 0} = id_{T_xM}$. 
    \end{itemize}
    Henceforth, we shall denote the collection of such maps by $\mathcal{Z}$. 
\end{definition}
\noindent In the aforementioned scenario, we shall adopt the following notation. Given $u \in \mathcal{D}'(M)$, we denote by $u^{\alpha} := (\textbf{1} \otimes u) \in \mathcal{D}'(M)$ the tensor product distribution whose action is defined on the space $\mathcal{D}_1(M) \otimes \mathcal{D}_1(M)$, where $\mathcal{D}_1(M) :=\{\psi \in \mathcal{D}(M) \, \vert \, \partial^{\alpha} f(0) = 0, \forall |\alpha| \ge 1\}$. Furthermore, for any $\lambda \in [0,1]$, $u^{\alpha}_{\lambda}(x,v) := u^{\alpha} (x, \lambda v)$ is a distribution on $\mathcal{D}_1(V)$. \\
We can now state an auxiliary definition which allows us to introduce the notion of \emph{microlocal scaling degree} at a submanifold. 
\begin{definition}[Conormal bundle]
    Let $M$ be a $n-$dimensional smooth manifold and let $N \subset M$ be an embedded submanifold with associated embedding $\iota: N \hookrightarrow M$. We call \emph{conormal bundle} to $N$ the set
    \begin{equation}
        \mathcal{N}^* N := \left \{ (x, \xi) \in T^*M \, \vert \, x \in N, \langle \xi, v \rangle = 0, \forall v \in T_x^{(M)} N \right\}, 
    \end{equation}
    where $v \in T_x^{(M)}N$ if and only if $\exists \overline{v} \in T_xN$ such that $\overline{v} = d \iota(v)$.
\end{definition}
\begin{definition}[Microlocal scaling degree]\label{Def: Microlocal scaling degree}
    Let $\alpha \in \mathcal{Z}$. A distribution $u \in \mathcal{D}'(M)$ has a \emph{microlocal scaling degree} $\mu \, sd_N^{\Gamma_0} (u, \alpha)$ at a submanifold $N$ with respect to the closed conical set $\Gamma_0 \subseteq \mathring{\mathcal{N}}^*N$ if 
    \begin{itemize}
        \item[1.] there exists $\Gamma \subset \mathring{T}^*(T_N M)$ such that $WF(u^{\alpha}_{\lambda}) \subset \Gamma$ and $\Gamma \vert_{TN} \subset \alpha^* (Z(T^*N) \times \Gamma_0)$, 
        \item[2.] it holds true that 
        \begin{equation}
            \mu \, sd^{\Gamma_0}_N (u, \alpha) = \inf \left \{ \omega \in \mathbb{R} \, \vert \, \lim_{\lambda \rightarrow 0^+} \lambda^{\omega} u^{\alpha}_{\lambda} = 0 \right \}, 
        \end{equation}
        where the limit is to be understood in the sense of the H\"ormander topology. 
    \end{itemize}
\end{definition}
\begin{remark}
    In the above definition, the choice of the closed conical neighbourhood is of key importance. 
\end{remark}
\begin{definition}
    Let $M$ be a smooth, $n-$dimensional manifold and let $N$ be an embedded submanifold. Suppose that $u \in \mathcal{D}'(M \setminus N)$. Then, for every admissible $\alpha \in \mathcal{Z}$, we define 
    \begin{equation}
        u^{\chi}_{\lambda} := (\textbf{1} \otimes \chi) \circ \alpha \cdot u^{\alpha}_{\lambda} \in \mathcal{D}'(V),
    \end{equation}
    where $\chi \in C^{\infty}(M), \, supp(\chi) \cap N = \emptyset$ and $V$ is a star-shaped neighbourhood of $Z(T_N M)$, while $\textbf{1}$ is the distribution generated by the constant function $1$. To $u$ we can associate its \emph{scaling degree at $N$} defined as follows:
    \begin{equation}
    \label{sdsub}
        sd_N(u) := \inf \left\{ \omega \in \mathbb{R} \, \vert \, \lim_{\lambda \rightarrow 0^+} \lambda^{\omega} u^{\chi}_{\lambda} = 0 \right\}. 
    \end{equation}
    Equation \eqref{sdsub} is independent of the choice of the function $\chi$. 
\end{definition}
We can now state a theorem -- see \cite[Thm. 6.9]{brunetti} -- providing a sufficient condition for the existence and uniqueness of extensions of a given distribution ill-defined at a submanifold. 
\begin{theorem}\label{Thm: Degree of divergence}
    Consider a smooth, $n-$dimensional manifold $M$ and an embedded submanifold $N$. let $u_0 \in \mathcal{D}'(M \setminus N)$. Then, 
    \begin{itemize}
        \item[1.] if $sd_n(u_0) < codim(N)$, then $\exists ! u \in \mathcal{D}'(M)$ such that 
        \begin{equation*}
            u \vert_{M \setminus N} = u_0, \, \, \, sd_N(u_0) = sd_N(u).
        \end{equation*}
        \item[2.] if $\infty > sd_N(u_0) \ge codim(N)$, then $\exists u \in \mathcal{D}'(M)$ such that 
        \begin{equation*}
            u \vert_{M \setminus N} = u_0, \, \, \, sd_N(u_0) = sd_N(u).
        \end{equation*}
        Moreover, the difference between any pair of admissible extensions $u, u'$ reads
        \begin{equation*}
            u - u' = \sum_{j=0}^{\rho_N(u)} c_j \delta^{(j)}_N, \, \, c_j \in C^{\infty}(N),
        \end{equation*}
        where $\rho_N(u) := \lfloor codim(N) - sd_N(u) \rfloor $ is the \emph{degree of divergence of $u$ at $N$}.  
    \end{itemize}
\end{theorem}
\end{appendices}


\begin{thebibliography}{42}
\addcontentsline{toc}{section}{References}







\bibitem[BK18]{Bahns:2016zqj}
D.~Bahns and K.~Rejzner,
{\em ``The Quantum Sine Gordon model in perturbative AQFT,''}
Comm. Math. Phys. \textbf{357} (2018) no.1, 421-446, [arXiv: 1609.08530 [math-ph]].

\bibitem[BPR23]{Bahns:2021klc}
D.~Bahns, N.~Pinamonti and K.~Rejzner,
{\em ``Equilibrium states for the massive Sine-Gordon theory in the Lorentzian signature,''}
J. Math. Anal. Appl. \textbf{526} (2023), 127249
[arXiv:2103.09328 [math-ph]].

\bibitem[BGP07]{bar}
C.~B\"ar, N.~Ginoux and F.~Pf\"affe, 
Wave Equations on Lorentzian Manifolds and Quantization,
{\it European Mathematical Society} (2007), 194 p.


\bibitem[BDR23a]{Bonicelli:2021uxu}
A.~Bonicelli, C.~Dappiaggi and P.~Rinaldi,
{\em ``An Algebraic and Microlocal Approach to the Stochastic Nonlinear Schr\"odinger Equation''},
Ann. Henri Poinc. \textbf{24} (2023) no.7, 2443-2482, [arXiv:2111.06320 [math-ph]].

\bibitem[BDR23b]{bonicelli}
A.~Bonicelli, C.~Dappiaggi and P.~Rinaldi, 
{\em ``On the stochastic Sine-Gordon model: an interacting field theory approach"}, 
[arXiv: 2311.01558 [math-ph]], to appear on Comm. Math. Phys. 

\bibitem[BF00]{brunetti}
R.~Brunetti and K.~Fredenhagen, 
{ \em ``Microlocal analysis and Interacting Quantum Field Theories: Renormalization on Physical Backgrounds"},
Comm. Math. Phys. \textbf{208}(3) (2000), 623-661, [arXiv:9903028 [math-ph]].

\bibitem[BDF09]{brunetti1}
R.~Brunetti, M.~Duetsch and K.~Fredenhagen, 
{\em ``Perturbative Algebraic Quantum Field Theory and the Renormalization Groups"},
Adv. Theor. Math. Phys. \textbf{13} (2009) no.5, 1541-1599, [arXiv:0901.2038 [math-ph]].

\bibitem[BDFY15]{brunetti2}
R.~Brunetti, C.~Dappiaggi, K.~Fredenhagen and J.~Yngvason, 
Advances in algebraic quantum field theory, 
{\it Springer} (2015), 455 p.

\bibitem[CCJ70]{coleman}
C.~G. Callan Jr., S.~Coleman and R.~Jackiw, 
{ \em ``A New Improved Energy-Momentum Tensor"},
Ann. of Phys. {\bf 59} (1970) no. 1, 42-73.

\bibitem[DHP09]{Dappiaggi:2009xj}
C.~Dappiaggi, T.~P.~Hack and N.~Pinamonti,
{\em ``The Extended algebra of observables for Dirac fields and the trace anomaly of their stress-energy tensor,''}
Rev. Math. Phys. \textbf{21}, 1241-1312 (2009)
[arXiv:0904.0612 [math-ph]].

\bibitem[DF08]{decanini}
Y.~D\'ecanini and A.~Folacci, 
{ \em ``Hadamard renormalization of the stress-energy tensor for a quantized scalar field in a general spacetime of arbitrary dimension"}, 
{\it Phys. Rev. D} \textbf{78} (2008), 044025, [arXiv:gr-qc/0512118 [gr-qc]].


\bibitem[DHP17]{drago}
N.~Drago, T.P.~Hack and N.~Pinamonti, 
{ \em ``The generalised principle of perturbative agreement and the thermal mass"}, Ann. Henri Poincaré, \textbf{18} (2017), no.3, 807-868, [arXiv:1502.02705 [math-ph]].

\bibitem[DF01]{deutsch}
M.~Duetsch and K.~Fredenhagen,
{ \em ``Perturbative algebraic field theory and deformation quantization"}, 
Fields Inst. Comm. \textbf{30} (2001), 151-160,
[arXiv:hep-th/0101079 [hep-th]].

\bibitem[DH71]{duist}
J.J.~Duistermaat and L.~H\"ormander, 
Fourier Integral Operators II, 
Acta Math., \textbf{127} (1971), 79-183.


\bibitem[EG73]{epstein}
H.~Epstein and V.~Glaser, 
{\em ``The role of locality in perturbation theory"}, 
Ann. Henri Poincaré, \textbf{19},  (1973), no.3, 211-295.

\bibitem[FFFL24]{Ferrero:2023unz}
R.~Ferrero, S.~A.~Franchino-Vi\~nas, M.~B.~Fr\"ob and W.~C.~C.~Lima,
{\em ``Universal Definition of the Nonconformal Trace Anomaly,''}
Phys. Rev. Lett. \textbf{132}, no.7, 071601 (2024)
[arXiv:2312.07666 [hep-th]].

\bibitem[FV13]{Fewster:2013lqa}
C.~J.~Fewster and R.~Verch,
{\em ``The Necessity of the Hadamard Condition,''}
Class. Quant. Grav. \textbf{30} (2013), 235027
[arXiv:1307.5242 [gr-qc]].

\bibitem[FR12]{Fredenhagen:2011an}
K.~Fredenhagen and K.~Rejzner,
{\em ``Batalin-Vilkovisky formalism in the functional approach to classical field theory,''}
Comm. Math. Phys. \textbf{314} (2012), 93-127
[arXiv:1101.5112 [math-ph]].

\bibitem[FR13]{Fredenhagen:2011mq}
K.~Fredenhagen and K.~Rejzner,
{\em ``Batalin-Vilkovisky formalism in perturbative algebraic quantum field theory,''}
Comm. Math. Phys. \textbf{317} (2013), 697-725
[arXiv:1110.5232 [math-ph]].

\bibitem[Fri10]{Friedlander:2010eqa}
F.~G.~Friedlander,
The Wave Equation on a Curved Space-Time,
{\it Cambridge University Press} (2010), 290p.

\bibitem[FZ19]{Frob:2019dgf}
M.~B.~Fr\"ob and J.~Zahn,
{\em ``Trace anomaly for chiral fermions via Hadamard subtraction,''}
JHEP \textbf{10}, 223 (2019)
[arXiv:1904.10982 [hep-th]].








\bibitem[Hat82]{Hathrell:1981zb}
S.~J.~Hathrell,
{\em ``Trace Anomalies and $\lambda \phi^4$ Theory in Curved Space,''}
Annals Phys. \textbf{139} (1982), 136



\bibitem[HW01]{hollands}
S.~Hollands and R.M.~Wald, 
{\em ``Local Wick Polynomials and Time Ordered Products of Quantum Fields in Curved Spacetime"}, 
Comm. Math. Phys., \textbf{223} (2001), 289-326, [arXiv:gr-qc/0103074 [gr-qc]]. 

\bibitem[HW02]{hollands1}
S.~Hollands and R.M.~Wald, 
{\em ``Existence of Local Covariant Time Ordered Products of Quantum Fields in Curved Spacetime"}, 
Comm. Math. Phys. \textbf{231} (2002), 309-345, [arXiv:gr-qc/0111108 [gr-qc]].

\bibitem[HW05]{hollands2}
S.~Hollands and R.M.~Wald, 
{\em ``Conservation of the Stress-Energy Tensor in Perturbative Interacting Quantum Field Theory in Curved Spacetimes"}, 
Rev. Math. Phys. \textbf{17} (2005), 227-312, [arXiv:gr-qc/0404074 [gr-qc]].


\bibitem[H\"or03]{horm1}
L.~H\"ormander, 
The Analysis of Linear Partial Differential Operators, Vol. 1: Distribution Theory and Fourier Analysis,
Classics in Mathematics, Springer, (2003), p. 440.

\bibitem[HS17]{hudson}
J.~Hudson and P.~Schweitzer, 
{\em ``$D$ term and the structure of pointlike and composed spin-0 particles"}, 
Phys. Rew. D, \textbf{96}, 11, 114013 (2017).

\bibitem[HS18]{hudson2}
J.~Hudson and P.~Schweitzer, 
{\em ``Dynamic origins of fermionic $D$-terms"},
Phys. Rew. D, \textbf{97}, 5, 056003 (2018).

\bibitem[KM14]{khav}
I.~Khavkine and V.~Moretti, 
{\em ``Algebraic QFT in Curved Space-time and Quasifree Hadamard States: An Introduction"},
Advances in Algebraic Quantum Field Theory, 191-251,
[arXiv:1412.5945 [math-ph]].



\bibitem[May24]{maynard}
B.~Maynard, 
{\em ``Energy-momentum tensor in $\Phi^4$ theory at one loop"}, 
[arXiV: 2406.08857 [hep-ph]].

\bibitem[Mor03]{moretti}
V.~Moretti, 
{\em ``Comments on the Stress-Energy Tensor Operator in Curved Spacetime"}, 
Comm. Math. Phys., \textbf{232} (2003), 189-221, [arXiv:gr-qc/0109048 [gr-qc]].


\bibitem[PS18]{poly}
M.V.~Polyakov and P.~Schweitzer, 
{\em ``Forces inside hadrons: pressure, surface tension, mechanical radius, and all that"}, 
Int. J. Mod. Phys. A, \textbf{33}, (2018), no.26, [arXiv:1805.06596 [hep-ph]].

\bibitem[Rad96a]{radzi}
M. J.~Radzikowski, 
{\em ``Micro-local approach to the Hadamard condition in quantum field theory on curved space-time"}, 
Comm. Math. Phys., \textbf{179}, (1996), 529-553. 

\bibitem[Rad96b]{radzi2}
M. J.~Radzikowski and R.~Verch, 
{\em ``A local-to-global singularity theorem for quantum field theory on curved space-time"}, 
Comm. Math. Phys. 180, (1996), 1-22.

\bibitem[Rej16]{rejzner}
K.~Rejzner, 
Perturbative Algebraic Quantum Field Theory, 
{\it Springer}, (2016), p. 180. 




\bibitem[Wal78]{wald}
R.M.~Wald, 
{\em ``Trace Anomaly of a conformally invariant quantum field in curved spacetime"}, 
Phys. Rev. D, \textbf{17}, 6, (1978), 1477-1484. 









\end{thebibliography}
\end{document}